\documentclass[preprint, 3p, times]{elsarticle}
\usepackage{custom-preamble}

\begin{document}

\begin{frontmatter}

\title{Reconstructing semi-directed level-1 networks using few quarnets\tnoteref{fund1}}

\author[inst1]{Martin Frohn}
\ead{martin.frohn@maastrichtuniversity.nl}

\author[inst2]{Niels Holtgrefe\corref{cor1}}
\ead{n.a.l.holtgrefe@tudelft.nl}

\author[inst2]{Leo van Iersel}
\ead{l.j.j.vaniersel@tudelft.nl}

\author[inst2]{Mark Jones}
\ead{m.e.l.jones@tudelft.nl}

\author[inst1]{Steven Kelk}
\ead{steven.kelk@maastrichtuniversity.nl}

\affiliation[inst1]{
            organization={Department of Advanced Computing Sciences, Maastricht University},
            addressline={Paul-Henri~Spaaklaan~1},
            city={Maastricht},
            postcode={6229~EN}, 
            country={The Netherlands}
            }
\affiliation[inst2]{
            organization={Delft Institute of Applied Mathematics, Delft University of Technology},
            addressline={Mekelweg~4}, 
            city={Delft},
            postcode={2628~CD}, 
            country={The Netherlands}
            }
            
\tnotetext[fund1]{This work was supported by grants OCENW.M.21.306 and OCENW.KLEIN.125 of the Dutch Research Council (NWO).}

\cortext[cor1]{Corresponding author}

\begin{abstract}
Semi-directed networks are partially directed graphs that model evolution where the directed edges represent reticulate evolutionary events. We present an algorithm that reconstructs binary $n$-leaf semi-directed level-1 networks in $\bigO( n^2)$ time from its quarnets (4-leaf subnetworks). Our method assumes we have direct access to all quarnets, yet uses only an asymptotically optimal number of $\bigO(n \log n)$ quarnets. When the network is assumed to contain no triangles, our method instead relies only on four-cycle quarnets and the splits of the other quarnets. A variant of our algorithm works with quartets rather than quarnets and we show that it reconstructs most of a semi-directed level-1 network from an asymptotically optimal $\bigO(n \log n)$ of the quartets it displays. Additionally, we provide an $\bigO(n^3)$ time algorithm that reconstructs the tree-of-blobs of any binary $n$-leaf semi-directed network with unbounded level from $\bigO(n^3)$ splits of its quarnets.
\end{abstract}

\begin{keyword}
Phylogenetics \sep Semi-directed network \sep Quarnets \sep Quartets \sep Tree-of-blobs \sep Splits
\end{keyword}

\end{frontmatter}

\section{Introduction}
Phylogenetic networks are directed acyclic graphs that model the evolutionary history of taxa, represented by the leaves of the graph. Unlike phylogenetic trees, these networks can account for reticulate events such as hybridization and horizontal gene transfer, making them essential tools in evolutionary studies \cite{bapteste2013networks}. Among these, phylogenetic \emph{level-1} networks (originally called `galled trees' \cite{wang2001perfect,gusfield2003efficient}) form a fundamental class \cite{kong2022classes}, characterized by tree-like components and isolated cycles with a single reticulation. Sol\'is-Lemus and An\'e \cite{solis2016inferring} introduced the semi-directed topology of a phylogenetic network (see \cref{fig:preliminaries}): partially directed graphs that only model the direction of evolutionary change for reticulate events while the location of the root of the network remains unknown. If an outgroup (i.e. a more distantly related taxon) is included in the analysis, the directed topology can be recovered from the semi-directed network \cite{solis2016inferring}. At the expense of containing less evolutionary information than their directed counterparts, semi-directed level-1 networks show favourable identifiability results under many evolutionary models \cite{banos2019identifying,solis2020identifiability,xu2023identifiability,allman2024identifiability,gross2018distinguishing,hollering2021identifiability,gross2021distinguishing}. That is, under these models, it is theoretically possible to infer (most of) the semi-directed level-1 network from different types of biological data (such as aligned nucleotide sequences or gene trees).

A major challenge in phylogenetics is the reliable construction of networks that best fit the finite biological data that is available. Various existing practical software tools can be used to infer phylogenetic level-1 networks (both directed and semi-directed) from biological data under different model assumptions and using different techniques. For example, \textsc{PhyloNet} \cite{than2008phylonet,yu2015maximum}, \textsc{SNaQ} \cite{solis2016inferring,solis2017phylonetworks} and \textsc{PhyNEST} \cite{kong2022inference} first build a starting network and then explore the space of networks to choose a best network based on a likelihood criterion. On the other hand, \textsc{Lev1athan} \cite{huber2010practical}, \textsc{TriLoNet} \cite{oldman2016trilonet}, and \textsc{NANUQ} \cite{allman2019nanuq} utilize information on substructures with few leaves to puzzle together the structure of the full network. Recently, two new programs that combine 4-leaf substructures to construct a semi-directed level-1 network were introduced: \textsc{NANUQ$^+$} \cite{NANUQplus} and \textsc{Squirrel} \cite{holtgrefe2024squirrel}. Several theoretical contributions focused on building phylogenetic level-1 networks from substructures have also appeared, see e.g. \cite{warnow2024advances, huebler2019constructing,gambette2012quartets,keijsper2014reconstructing, huber2018quarnet} and references therein.

A common starting point in many of these methods is to first infer a set of unrooted displayed \emph{quartets} (4-leaf trees) for every subset of four leaves \cite{solis2016inferring,allman2019nanuq,NANUQplus,warnow2024advances,gambette2012quartets,keijsper2014reconstructing}. Note that the displayed quartets of a semi-directed network are the 4-leaf trees displayed in the 4-leaf subnetworks of the network. Ba{\~n}os \cite{banos2019identifying} showed that the displayed quartets of a semi-directed level-1 network do not provide enough information to reconstruct every possible semi-directed level-1 network. On the positive side, Huber et al. \cite{huber2024splits} recently showed that all semi-directed level-1 (and even level-2) networks can be distinguished by the semi-directed \emph{quarnets} (4-leaf subnetworks) they induce. Motivated by positive identifiability results for inferring quarnets under certain group-based models of evolution (Jukes-Cantor, Kimura-2 and Kimura-3) \cite{gross2018distinguishing,hollering2021identifiability,gross2021distinguishing}, some practical programs have already arisen that infer such quarnets from practical sequence data \cite{holtgrefe2024squirrel,barton2022statistical,martin2023algebraic}. These programs have inspired the software tool \textsc{Squirrel} \cite{holtgrefe2024squirrel}, which employs quarnets without triangles (3-cycles) and then builds a semi-directed triangle-free level-1 network. Additionally, Huebler et al. \cite{huebler2019constructing} described two theoretical algorithms that reconstruct an $n$-leaf semi-directed level-1 network using all of its $\Theta (n^4)$ induced quarnets, although full correctness was not formally proved. This is in stark contrast with the best algorithms that reconstruct binary phylogenetic trees from quartets. Assuming direct access to all quartets is available, this can be done in $\bigO (n \log n)$ time, while using only $\bigO (n \log n)$ quartets \cite{lingas1999efficient,brodal2001complexity}. Here, assuming direct access means that an `oracle' is available from which the algorithm can query a quartet on any set of leaves.




Reducing the computational overhead by utilizing only the essential quarnets could significantly speed up practical network construction algorithms, enabling more efficient and scalable approaches to reconstructing semi-directed level-1 networks. We present an algorithm that constructs a semi-directed level-1 network from its induced quarnets with bounds similar to those for tree reconstruction, thus taking a first step towards network construction methods that rely only on a small subset of quarnets. Specifically, our algorithm reconstructs an $n$-leaf semi-directed level-1 network using $\bigO (n \log n)$ quarnets (again assuming direct access to all quarnets), and running in $\bigO (n^2)$ time. Similar to trees, we show that one needs $\Omega (n \log n)$ quarnets, thus proving optimality of the number of quarnets that are used.

Our algorithm works by repeatedly attaching leaves to construct a canonical form of a network while using only the splits of the quarnets. Subsequently, the four-cycle quarnets are used to determine the remaining parts of the larger cycles, while the triangles in this canonical form can be inferred from the triangles in the quarnets. The structure of our algorithm aligns well with the earlier mentioned existing quarnet inference methods which provide evidence that inferring the splits of quarnets from data is easier than inferring four-cycle quarnets, while identifying triangles in quarnets is the most challenging \cite{holtgrefe2024squirrel,barton2022statistical,martin2023algebraic}. Having access to the displayed quartets of a network allows one to deduce most of the topology of the quarnets of the network \cite{banos2019identifying}. Consequently, our approach also provides similar complexity improvements for the related problem of constructing (most of) a semi-directed level-1 network from its displayed quartets.





Additionally, this paper covers the \emph{tree-of-blobs} of a semi-directed network, originally introduced by Gusfield et al. \cite{gusfield2007decomposition} for directed phylogenetic networks. Such a tree shows only the tree-like aspects of a network and was shown to be identifiable under several models \cite{xu2023identifiability,allman2023tree,rhodes2024identifying}. The tree-of-blobs is for example useful in the context of \cite{rhodes2024identifying}. There it was shown that under several evolutionary models, once the tree-of-blobs is known, displayed quartets can be used to infer the circular order in which subnetworks attach to an outer-labeled planar blob. On the algorithmic side, Allman et al. \cite{allman2024tinnik} recently introduced the practical quartet-based program \textsc{TINNiK} that constructs a tree-of-blobs from gene trees under the network multispecies coalescent model. Gambette, Berry, and Paul \cite{gambette2012quartets} presented an algorithm to construct the `SN-tree' from a set of quartets, which can also be used to construct the tree-of-blobs of an $n$-leaf semi-directed network from all $\Theta (n^4)$ quartets displayed by the network. Furthermore, the \textsc{IQ$^*$} algorithm \cite{berry2000inferring} can be applied to construct the tree-of-blobs using all $\Theta (n^4)$ splits of the quarnets induced by a semi-directed network (see the proof in the supplementary material of \cite{holtgrefe2024squirrel}). In our paper, we improve upon these aforementioned algorithms by reconstructing the tree-of-blobs of any $n$-leaf semi-directed network in $\bigO (n^3)$ time using $\bigO (n^3)$ splits of its quarnets, which can also be deduced from the displayed quartets of the network.

The remainder of the paper is structured as follows. \cref{sec:preliminaries} contains definitions and notation, while \cref{sec:splits} introduces some crucial results on splits in semi-directed networks. We continue with the tree-of-blobs reconstruction algorithm in \cref{sec:blobtree} and refine our approach in \cref{sec:level-1} to reconstruct level-1 networks. We end with a discussion in \cref{sec:discussion}.

\section{Preliminaries}\label{sec:preliminaries}

\paragraph{Phylogenetic trees and networks}
A \emph{blob} of a directed graph is a maximal subgraph without any cut-edges and it is an \emph{$m$-blob} if it has $m$ edges incident to it. A \emph{(binary) directed phylogenetic network} on a set of at least two leaves $\X$ is a rooted directed acyclic graph with no edges in parallel such that (i): it has no 1-blobs or 2-blobs, other than possibly a blob with no incoming and two outgoing edges; (ii): the root has out-degree two; (iii): each vertex with out-degree zero has in-degree one and the set of such vertices is $\X$; (iv): all other vertices either have in-degree one and out-degree two, or in-degree two and out-degree one. A vertex of the last type is a \emph{reticulation vertex}, and the two edges directed towards it are called \emph{reticulation edges}. A directed phylogenetic network without reticulation vertices is a \emph{(binary) directed phylogenetic tree}. The type of network this paper revolves around can be obtained from a directed phylogenetic network as follows.
\begin{definition}[Semi-directed network]\label{def:semi_directed}
A \emph{(binary) semi-directed phylogenetic network} $\N$ on $\X$ is a partially directed graph that (i): can be obtained from a directed phylogenetic network by undirecting all non-reticulation edges and suppressing the former root; (ii): does not have any parallel edges, degree-2 vertices, 1-blobs or 2-blobs.
\end{definition}
For the sake of brevity, we refer to the above networks simply as \emph{semi-directed networks}. Since we do not undirect reticulation edges, we can still refer to the reticulation vertices and edges of a semi-directed network. The \emph{skeleton} of a semi-directed network is the undirected graph that is obtained from the network by removing all edge directions. A \emph{(binary undirected) phylogenetic tree} on at least two leaves $\X$ is an undirected binary tree with leaf set $\X$ and no degree-2 vertices. Clearly, such a tree is a specific type of semi-directed network. The non-leaf vertices of such a tree are \emph{internal vertices}. We say that a phylogenetic tree is a \emph{displayed tree} of a semi-directed network on $\X$ if the tree can be obtained from the network by deleting one reticulation edge per reticulation vertex, undirecting the remaining reticulation edges, removing all vertices not on any path between a pair of leaves in $\X$, and suppressing degree-2 vertices.

Analogous to a directed phylogenetic network, a \emph{blob} of a semi-directed network is a maximal subgraph without any cut-edges and it is an \emph{$m$-blob} or has \emph{degree-$m$} if it has $m$ edges incident to it. The blob is \emph{trivial} if it consists of a single vertex and it is \emph{internal} if it is not a leaf. A semi-directed network is \emph{simple} if it has at most one internal blob, it is \emph{level-$\ell$} if every blob contains at most $\ell$ reticulation vertices, and it is \emph{strict level-$\ell$} if it is not level-$(\ell-1)$. Consequently, a semi-directed network is level-1 if every blob is either a single vertex or a \emph{$k$-cycle}: a cycle of $k\geq 3$ vertices where directions are disregarded. We refer to 3-cycles as \emph{triangles} and call a semi-directed network \emph{triangle-free} if none of its blobs are triangles.

The \emph{tree-of-blobs} $\T ( \N )$ of a semi-directed network $\N$ is obtained by contracting every blob $\B$ to a single vertex $v$, in which case we say that \emph{$v$ represents $\B$}. By definition, the tree-of-blobs $\T (\N)$ must be an undirected phylogenetic tree on $\X$ since a semi-directed network contains no 2-blobs. It is well-known that constructing the tree-of-blobs of a given network takes linear time \cite{hopcroft1973algorithm}. To illustrate some of the previous definitions, we refer to \cref{fig:preliminaries}, which shows a semi-directed level-2 network and its tree-of-blobs.

\begin{figure}[htb]
\centering
\begin{tikzpicture}[scale=0.5]
	\begin{pgfonlayer}{nodelayer}
		\node [style={leaf_node}, label={above:1}] (0) at (2.25, 7.75) {};
		\node [style={leaf_node}, label={left:2}] (1) at (-0.75, 4.75) {};
		\node [style={leaf_node}, label={below:3}] (2) at (2.25, 1.75) {};
		\node [style={leaf_node}, label={below:4}] (3) at (4.75, 0.5) {};
		\node [style={leaf_node}, label={below:5}] (4) at (6.5, 2.25) {};
		\node [style={leaf_node}, label={below:6}] (5) at (8, 2.75) {};
		\node [style={leaf_node}, label={right:7}] (6) at (10, 4.75) {};
		\node [style={leaf_node}, label={above right:8}] (7) at (9, 7.75) {};
		\node [style={leaf_node}, label={above left:9}] (8) at (7, 7.75) {};
		\node [style={leaf_node}, label={above:10}] (9) at (4.5, 7) {};
		\node [style={internal_node}] (16) at (2.25, 6.75) {};
		\node [style={internal_node}] (17) at (2.25, 2.75) {};
		\node [style={internal_node}] (18) at (0.25, 4.75) {};
		\node [style={internal_node}] (19) at (4.25, 4.75) {};
		\node [style={internal_node}] (20) at (0.75, 6.25) {};
		\node [style={internal_node}] (21) at (0.75, 3.25) {};
		\node [style={internal_node}] (22) at (3.75, 3.25) {};
		\node [style={internal_node}] (23) at (3.75, 6.25) {};
		\node [style={internal_node}] (24) at (4.5, 2.5) {};
		\node [style={internal_node}] (25) at (4.75, 1.5) {};
		\node [style={internal_node}] (26) at (5.5, 2.25) {};
		\node [style={internal_node}] (27) at (7, 4.75) {};
		\node [style={internal_node}] (28) at (8, 5.75) {};
		\node [style={internal_node}] (29) at (8, 3.75) {};
		\node [style={internal_node}] (30) at (9, 4.75) {};
		\node [style={internal_node}] (31) at (8, 6.75) {};
		\node [style={main_label}] (32) at (-0.25, 8.5) {$\N$};
		\node [style={internal_node}] (33) at (21.5, 4.5) {};
		\node [style={internal_node}] (34) at (21.5, 3.5) {};
		\node [style={internal_node}] (35) at (22.25, 4) {};
		\node [style={internal_node}] (36) at (23.5, 4) {};
		\node [style={internal_node}] (37) at (24.75, 4) {};
		\node [style={internal_node}] (38) at (24.75, 3.25) {};
		\node [style={internal_node}] (39) at (24.25, 2.5) {};
		\node [style={internal_node}] (40) at (25.25, 2.5) {};
		\node [style={internal_node}] (41) at (26, 4) {};
		\node [style={leaf_node}, label={left:1}] (42) at (20.75, 5.25) {};
		\node [style={leaf_node}, label={left:2}] (43) at (20.75, 2.75) {};
		\node [style={leaf_node}, label={below:3}] (44) at (23.5, 3.25) {};
		\node [style={leaf_node}, label={below:4}] (45) at (23.75, 1.75) {};
		\node [style={leaf_node}, label={below:5}] (46) at (25.75, 1.75) {};
		\node [style={leaf_node}, label={right:6}] (47) at (26.75, 3.25) {};
		\node [style={leaf_node}, label={right:7}] (48) at (26.75, 4.75) {};
		\node [style={main_label}] (49) at (24.25, 7) {$\N|_{\X_7}$};
		\node [style={internal_node}] (50) at (14.5, 4) {};
		\node [style={internal_node}] (51) at (16.5, 4) {};
		\node [style={leaf_node}, label={above:9}] (52) at (16, 5.75) {};
		\node [style={leaf_node}, label={above:8}] (53) at (17, 5.75) {};
		\node [style={leaf_node}, label={right:7}] (54) at (17.75, 4) {};
		\node [style={leaf_node}, label={below:6}] (55) at (16.5, 2.75) {};
		\node [style={leaf_node}, label={above:10}] (56) at (14.5, 5.25) {};
		\node [style={leaf_node}, label={left:1}] (57) at (13.5, 5) {};
		\node [style={leaf_node}, label={left:2}] (58) at (13, 4) {};
		\node [style={leaf_node}, label={left:3}] (59) at (13.5, 3) {};
		\node [style={leaf_node}, label={below:4}] (60) at (14, 2.25) {};
		\node [style={leaf_node}, label={below:5}] (61) at (15, 2.25) {};
		\node [style={internal_node}] (62) at (16.5, 5) {};
		\node [style={internal_node}] (63) at (14.5, 3) {};
		\node [style={main_label}] (64) at (15.25, 8) {$\T (\N)$};
	\end{pgfonlayer}
	\begin{pgfonlayer}{edgelayer}
		\draw [style={ret_arc}, bend left] (20) to (21);
		\draw [style={ret_arc}, bend right=15] (18) to (21);
		\draw [style={ret_arc}, bend left=15] (16) to (23);
		\draw [style={ret_arc}, bend right=15] (19) to (23);
		\draw [style=edge, bend left=15] (19) to (22);
		\draw [style=edge, bend left=15] (22) to (17);
		\draw [style=edge, bend left=15] (17) to (21);
		\draw [style=edge, bend left=15] (18) to (20);
		\draw [style=edge, bend left=15] (20) to (16);
		\draw [style=edge] (1) to (18);
		\draw [style=edge] (0) to (16);
		\draw [style=edge] (17) to (2);
		\draw [style=edge] (23) to (9);
		\draw [style=edge] (22) to (24);
		\draw [style=edge] (25) to (3);
		\draw [style=edge] (26) to (4);
		\draw [style={ret_arc}] (24) to (25);
		\draw [style={ret_arc}] (26) to (25);
		\draw [style=edge] (24) to (26);
		\draw [style={ret_arc}, bend left] (27) to (28);
		\draw [style={ret_arc}, bend right] (30) to (28);
		\draw [style=edge] (29) to (5);
		\draw [style=edge, bend right] (29) to (30);
		\draw [style=edge] (30) to (6);
		\draw [style=edge] (28) to (31);
		\draw [style=edge] (31) to (7);
		\draw [style=edge] (8) to (31);
		\draw [style=edge] (27) to (19);
		\draw [style=edge, bend left] (29) to (27);
		\draw [style={ret_arc}] (33) to (35);
		\draw [style={ret_arc}] (34) to (35);
		\draw [style={ret_arc}] (38) to (39);
		\draw [style={ret_arc}] (40) to (39);
		\draw [style=edge] (39) to (45);
		\draw [style=edge] (40) to (46);
		\draw [style=edge] (40) to (38);
		\draw [style=edge] (38) to (37);
		\draw [style=edge] (37) to (41);
		\draw [style=edge] (41) to (48);
		\draw [style=edge] (47) to (41);
		\draw [style=edge] (36) to (37);
		\draw [style=edge] (36) to (44);
		\draw [style=edge] (35) to (36);
		\draw [style=edge] (34) to (33);
		\draw [style=edge] (33) to (42);
		\draw [style=edge] (34) to (43);
		\draw [style=edge] (56) to (50);
		\draw [style=edge] (50) to (57);
		\draw [style=edge] (58) to (50);
		\draw [style=edge] (50) to (59);
		\draw [style=edge] (50) to (63);
		\draw [style=edge] (63) to (60);
		\draw [style=edge] (61) to (63);
		\draw [style=edge] (50) to (51);
		\draw [style=edge] (51) to (55);
		\draw [style=edge] (51) to (54);
		\draw [style=edge] (62) to (51);
		\draw [style=edge] (62) to (53);
		\draw [style=edge] (62) to (52);
	\end{pgfonlayer}
\end{tikzpicture}
\caption{A semi-directed level-2 network $\N$ on $\X = \{1, \ldots, 10\}$, its tree-of-blobs $\T(\N)$, and its subnetwork $\N|_{\X_7}$ induced by the first seven leaves $\X_7 = \{1, \ldots, 7\}$.}
\label{fig:preliminaries}
\end{figure}

An \emph{up-down path} between two leaves $x_1$ and $x_2$ of a semi-directed network is a path of $k$ edges where the first $\ell$ edges are directed towards $x_1$ and the last $k - \ell$ edges are directed towards $x_2$. Here, we consider undirected edges to be bidirected. With this notion, we can define the (induced) subnetwork of a semi-directed network (illustrated by \cref{fig:preliminaries}). This was shown to be well-defined in \cite{huber2024splits}, while Ba{\~n}os \cite{banos2019identifying} showed the useful fact that first taking the semi-directed network of a directed network and then inducing a subnetwork is equivalent to first inducing a directed subnetwork and then making it a semi-directed network. In case the original network is a phylogenetic tree, we often call the subnetwork a \emph{subtree} instead.
\begin{definition}[Subnetwork]\label{def:restriction}
Given a semi-directed network $\N$ on $\X$ and some $\Y \subseteq \X$ with $| \Y | \geq 2$, the \emph{subnetwork} of $\N$ induced by $\Y$ is the semi-directed network $\N|_\Y$ obtained from $\N$ by taking the union of all up-down paths between leaves in $\Y$, followed by exhaustively suppressing all 2-blobs and degree-2 vertices, and identifying parallel edges.    
\end{definition}
An interesting fact is that, in general, $\T ( \N)|_\Y \neq \T (\N|_\Y)$. In other words, the subtree of a tree-of-blobs need not be equal to the tree-of-blobs of a subnetwork. In particular, the latter can be more refined. As an example, consider \cref{fig:preliminaries}, where we have that $\T ( \N)|_{\X_7} \neq \T (\N|_{\X_7})$. This property shows why reconstructing the tree-of-blobs of a semi-directed network is inherently different than doing so for undirected networks. The difficulty for semi-directed networks is that a cut-edge separating two sets of leaves in a subnetwork might not exist in the full network, which would be the case for undirected networks.

As a technical convention, whenever we consider a semi-directed network $\N$ on $\X$ and some $\Y \subseteq \X$, we implicitly assume that $|\Y| \geq 2$ such that $\N|_\Y$ is a valid semi-directed network. Similarly, whenever we consider some $\Y \subset \X$ with $x \in \X \setminus \Y$, we assume that $|\Y| \geq 2$ such that $|\X| \geq 3$.

\paragraph{Splits, quarnets, quarnet-splits and displayed quartets}

Given a semi-directed network $\N$ on $\X$ and a partition $A|B$ of $\X$ (with $A$ and $B$ both non-empty), we say that $A|B$ is a \emph{split} in $\N$ if there exists a cut-edge of $\N$ whose removal disconnects the leaves in $A$ from those in $B$. The split and the cut-edge are \emph{non-trivial} if the corresponding partition is non-trivial, that is, if $|A|, |B| \geq 2$. For example, $\{1, 2, 3, 4, 5, 10\} | \{6, 7, 8, 9\}$ is a non-trivial split in the semi-directed network $\N$ from \cref{fig:preliminaries}. We may sometimes omit the word `non-trivial' if it is clear from the context. A split $A|B$ is \emph{induced by the cut-edge $uv$} if $u$ (resp. $v$) is \emph{on the side of $A$} (resp. $B$), meaning that $u$ (resp. $v$) is in the component of $\N$ containing $A$ (resp. $B$) after removing $uv$. For splits with few leaves, we sometimes omit the set notation, e.g. $ab | cd$ instead of $\{a, b\} | \{c, d\}$. A very well-known result is that an undirected (possibly non-binary) phylogenetic tree is uniquely determined by its (non-trivial) splits \cite{buneman1971recovery}. Consequently, the tree-of-blobs of a semi-directed network is the unique undirected phylogenetic tree with the same set of splits as the network itself.

Given a semi-directed network $\N$ on $\X$ and four different leaves $a, b, c, d \in \X$, the \emph{(semi-directed) quarnet} of $\N$ on $\{ a, b, c, d \}$ is the subnetwork of $\N$ induced by $\{ a, b, c, d \}$. A semi-directed level-1 network can have six different types of quarnets, up to labeling the leaves: the \emph{quartet tree}, two \emph{single triangles} that share the same skeleton, two \emph{double triangles} that share the same skeleton, and the \emph{four-cycle} (which we distinguish from a 4-cycle in a network by writing out the number 4). \cref{fig:quarnets} shows these six quarnet types. Any quarnet $\N|_{\{a, b, c, d\}}$ has four trivial splits: one cutting off each leaf. Next to that, the quarnet either has no non-trivial split at all, or it has exactly one non-trivial split ($ab|cd$, $ac|bd$, or $ad|bc$). We say that $ab | cd$ is a \emph{quarnet-split} of $\N$ if $ab |cd$ is a split of the quarnet $\N|_{\{a, b, c, d\}}$. We call the set of quartet trees that are displayed by the quarnets of a semi-directed network the \emph{(displayed) quartets} of the semi-directed network. Note that there are $\Theta (n^4)$ quarnets and $\Theta (n^4)$ displayed quartets of a semi-directed network on $n$ leaves. An important observation we will often implicitly use is the following.
\begin{observation}\label{obs:restriction}
Let $\N$ be a semi-directed network on $\X$ with $\{a, b, c, d\} \subseteq \Y \subseteq \X$. Then, the quarnet of $\N$ on $\{a, b, c, d\}$ is equal to the quarnet of $\N|_\Y$ on $\{a, b, c, d\}$. That is, $\N|_{\{a, b, c, d\}} = (\N|_\Y)|_{\{a, b, c, d\}}$.
\end{observation}

\begin{figure}[htb]
\centering
\tikzstyle{main_label}=[none,scale=1.0]

\begin{tikzpicture}[scale=0.4]
	\begin{pgfonlayer}{nodelayer}
		\node [style={internal_node}] (0) at (2, 5) {};
		\node [style={internal_node}] (1) at (2, 3) {};
		\node [style={leaf_node}] (2) at (0.5, 6.5) {};
		\node [style={leaf_node}] (3) at (3.5, 6.5) {};
		\node [style={leaf_node}] (4) at (0.5, 1.5) {};
		\node [style={leaf_node}] (5) at (3.5, 1.5) {};
		\node [style={internal_node}] (6) at (7.75, 5.5) {};
		\node [style={internal_node}] (7) at (8.5, 4.5) {};
		\node [style={internal_node}] (8) at (9.25, 5.5) {};
		\node [style={internal_node}] (9) at (8.5, 3) {};
		\node [style={leaf_node}] (10) at (7, 6.5) {};
		\node [style={leaf_node}] (11) at (10, 6.5) {};
		\node [style={leaf_node}] (12) at (7, 1.5) {};
		\node [style={leaf_node}] (13) at (10, 1.5) {};
		\node [style={internal_node}] (14) at (11.75, 5.5) {};
		\node [style={internal_node}] (15) at (12.5, 4.5) {};
		\node [style={internal_node}] (16) at (13.25, 5.5) {};
		\node [style={internal_node}] (17) at (12.5, 3) {};
		\node [style={leaf_node}] (18) at (11, 6.5) {};
		\node [style={leaf_node}] (19) at (14, 6.5) {};
		\node [style={leaf_node}] (20) at (11, 1.5) {};
		\node [style={leaf_node}] (21) at (14, 1.5) {};
		\node [style={internal_node}] (22) at (18.25, 5.5) {};
		\node [style={internal_node}] (23) at (19, 4.5) {};
		\node [style={internal_node}] (24) at (19.75, 5.5) {};
		\node [style={leaf_node}] (25) at (17.5, 6.5) {};
		\node [style={leaf_node}] (26) at (20.5, 6.5) {};
		\node [style={internal_node}] (27) at (19, 3.5) {};
		\node [style={internal_node}] (28) at (18.25, 2.5) {};
		\node [style={internal_node}] (29) at (19.75, 2.5) {};
		\node [style={leaf_node}] (30) at (17.5, 1.5) {};
		\node [style={leaf_node}] (31) at (20.5, 1.5) {};
		\node [style={internal_node}] (32) at (22.25, 5.5) {};
		\node [style={internal_node}] (33) at (23, 4.5) {};
		\node [style={internal_node}] (34) at (23.75, 5.5) {};
		\node [style={leaf_node}] (35) at (21.5, 6.5) {};
		\node [style={leaf_node}] (36) at (24.5, 6.5) {};
		\node [style={internal_node}] (37) at (23, 3.5) {};
		\node [style={internal_node}] (38) at (22.25, 2.5) {};
		\node [style={internal_node}] (39) at (23.75, 2.5) {};
		\node [style={leaf_node}] (40) at (21.5, 1.5) {};
		\node [style={leaf_node}] (41) at (24.5, 1.5) {};
		\node [style={leaf_node}] (50) at (28, 6.25) {};
		\node [style={leaf_node}] (51) at (32, 6.25) {};
		\node [style={leaf_node}] (52) at (28, 1.75) {};
		\node [style={leaf_node}] (53) at (32, 1.75) {};
		\node [style={internal_node}] (54) at (29, 5) {};
		\node [style={internal_node}] (55) at (29, 3) {};
		\node [style={internal_node}] (56) at (31, 3) {};
		\node [style={internal_node}] (57) at (31, 5) {};
		\node [style={main_label}] (58) at (2, 0.25) {Quartet tree};
		\node [style={main_label}] (59) at (10.5, 0.25) {Single triangle};
		\node [style={main_label}] (60) at (21, 0.25) {Double triangle};
		\node [style={main_label}] (61) at (30, 0.25) {Four-cycle};
	\end{pgfonlayer}
	\begin{pgfonlayer}{edgelayer}
		\draw (0) to (1);
		\draw (2) to (0);
		\draw (0) to (3);
		\draw (1) to (4);
		\draw (1) to (5);
		\draw (12) to (9);
		\draw (9) to (13);
		\draw (9) to (7);
		\draw (6) to (8);
		\draw (8) to (11);
		\draw (10) to (6);
		\draw (20) to (17);
		\draw (17) to (21);
		\draw (17) to (15);
		\draw (16) to (19);
		\draw (18) to (14);
		\draw (15) to (16);
		\draw [style={ret_arc}] (16) to (14);
		\draw [style={ret_arc}] (15) to (14);
		\draw [style={ret_arc}] (6) to (7);
		\draw [style={ret_arc}] (8) to (7);
		\draw (24) to (26);
		\draw (25) to (22);
		\draw [style=edge] (23) to (27);
		\draw [style=edge] (29) to (31);
		\draw [style=edge] (28) to (30);
		\draw (34) to (36);
		\draw (35) to (32);
		\draw [style=edge] (33) to (37);
		\draw [style=edge] (39) to (41);
		\draw [style=edge] (38) to (40);
		\draw [style={ret_arc}] (23) to (22);
		\draw [style={ret_arc}] (24) to (22);
		\draw [style={ret_arc}] (27) to (28);
		\draw [style={ret_arc}] (29) to (28);
		\draw [style={ret_arc}] (32) to (33);
		\draw [style={ret_arc}] (34) to (33);
		\draw [style={ret_arc}] (37) to (38);
		\draw [style={ret_arc}] (39) to (38);
		\draw [style=edge] (37) to (39);
		\draw [style=edge] (34) to (32);
		\draw [style=edge] (24) to (23);
		\draw [style=edge] (27) to (29);
		\draw [style=edge] (56) to (53);
		\draw [style=edge] (55) to (52);
		\draw [style=edge] (54) to (50);
		\draw [style=edge] (57) to (51);
		\draw [style={ret_arc}, bend right] (54) to (55);
		\draw [style={ret_arc}, bend left] (56) to (55);
		\draw [style=edge, bend right] (56) to (57);
		\draw [style=edge, bend left=330] (57) to (54);
	\end{pgfonlayer}
\end{tikzpicture}
\caption{The six possible quarnets of a semi-directed level-1 network, up to labeling the leaves (indicated by the filled black vertices).}
\label{fig:quarnets}
\end{figure}

The algorithms in this paper revolve around quarnet-splits and quarnets of semi-directed networks. We emphasize that from an informational perspective, having access to the displayed quartets of a semi-directed network lies somewhere in between these two. In particular, $ab|cd$ is a quarnet-split of a semi-directed network $\N$ if and only if $\N$ has exactly one displayed quartet on leaf set $\{a,b,c,d\}$ and this quartet has the split $ab|cd$ \cite[Lem.\,5.1]{rhodes2024identifying}. Therefore, the quarnet-splits of a semi-directed network can be deduced from its displayed quartets. Although the quarnet-splits might contain less information, they are possibly easier to infer from data \cite{martin2023algebraic}. Since there are at most three times as many displayed quartets compared to quarnet-splits, our asymptotic complexity results of the two algorithms that rely on quarnet-splits (\cref{alg:blobtree,alg:canonical_network_recursive}) also hold when using displayed quartets instead. Lastly, whenever we say that an algorithm `uses' $m$ quarnet-splits of $\N$, this means the algorithm examines $m$ quarnets to determine if they have a non-trivial split, and, if so, identifies the specific split. Thus, the count $m$ accounts for all quarnets analyzed for the presence of a split, even those turning out not to have a non-trivial split. We use this terminology to emphasize that only the (possible absence of a) non-trivial split of a quarnet is used, and not the complete topology of the quarnet.

\section{Determining splits of semi-directed networks with quarnets}\label{sec:splits}

This section is focused on a result that connects the splits in a semi-directed network to the quarnet-splits of that network, which will be useful in the algorithms in later sections. Recall that by \cref{obs:restriction} the quarnet(-split) of a subnetwork is equal to the quarnet(-split) of the whole network. Before proving the main result, we will first present \cref{lem:splits}, which is a slightly stronger version of a result by Huber et al. \cite{huber2024splits}. They show that $A|B$ is a split of a semi-directed phylogenetic network if and only if $a_1 a_2 | b_1 b_2$ is a quarnet-split for all $a_1, a_2 \in A$ and $b_1, b_2 \in B$, whereas we allow the leaves $a_1$ and $b_1$ to be fixed. The proof is a mere adjusted version of the proof by Huber et al. \cite{huber2024splits} and is thus deferred to the appendix.

\begin{restatable}{lemma}{splitlemma}\label{lem:splits}
Given a semi-directed network $\N$ on $\X$, a non-trivial partition $A|B$ of $\X$ and any $a_1 \in A$, $b_1 \in B$, the following are equivalent:
\begin{enumerate}[label={(\roman*)}, noitemsep,topsep=0pt]
\item $A|B$ is a split in $\N$;
\item $a_1 a_2 | b_1 b_2$ is a quarnet-split of $\N$ for all $a_2 \in A \setminus \{a_1 \}$, $b_2 \in B \setminus \{ b_1 \}$.
\end{enumerate}
\end{restatable}

Since an undirected phylogenetic tree is uniquely determined by its splits, this lemma also implies that the tree-of-blobs of a semi-directed network $\N$ can be distinguished by the quarnet-splits of $\N$. In other words, two networks with the same set of quarnet-splits must have the same tree-of-blobs. Consequently, we immediately get a slow $\bigO(2^n \cdot n^2)$ algorithm for free that reconstructs the tree-of-blobs from the quarnet-splits of $\N$: using \cref{lem:splits}, we check for all $\bigO (2^n)$ partitions of $\X$ whether it is a split of $\N$ (and thus of $\T(\N)$) in $\bigO(n^2)$ time. From these splits we can then build the tree-of-blobs.

Towards the next theorem we introduce a few notions, some of which are generalized from \cite{gross2021distinguishing}. Let $\B$ be an internal blob of a semi-directed network $\N$ on $\X$ and let $v$ be the internal vertex of $\T (\N)$ that represents it. Let $u_1 w_1, \ldots, u_s w_s$ be the cut-edges of $\N$ incident to $\B$ (with the $u_i$ in $\B$). The \emph{partition induced by $\B$} (or by $v$) is the partition $X_1| \ldots |X_s $ of $\X$ such that $x \in X_i$ if and only if $x$ is separated from $\B$ by $u_i w_i$. Given a reticulation vertex $r$ of $\B$, a set $X_i$ is \emph{below the reticulation $r$} if there exists a partially directed path from $r$ to $w_i$ and we then also say that $x$ is \emph{below the reticulation $r$} for any $x \in X_i$. A set $D \subseteq \X$ is a \emph{distinguishing set of $\B$} if it contains at least one leaf below each reticulation of $\B$ and it contains at least one leaf of three different sets $X_i$. As an example, consider the blob $\B$ of the semi-directed network $\N|_\Y$ from \cref{fig:split_theorem} (excluding the grey edge and the grey leaf $x$). The partition induced by this blob is $\{a_1, a_2, a_3, a_4\}|\{b_1\}|\{b_2\}|\{b_3, b_4\}|\{b_5\}|\{b_6\}$, while $b_3$ and $b_4$ are below one reticulation of $\B$ and $b_6$ is below the other reticulation. The sets $\{a_1, b_3, b_6\}$ and $\{b_1, b_4, b_6\}$ are examples of distinguishing sets of this blob.

\begin{figure}[htb]
\centering
\tikzstyle{special_edge3}=[-]
\tikzstyle{special_edge4}=[-, draw=red, thick]
\tikzstyle{special_edge5}=[draw=gray]
\tikzstyle{special_edge6}=[densely dashed, -{Latex[scale=.9]}, draw=gray]
\tikzstyle{special_edge7}=[densely dashed, draw=gray]
\tikzstyle{special_edge8}=[-, draw=red, ultra thick]
\tikzstyle{special_node2}=[circle, draw=gray, fill=gray, scale=0.275]
\tikzstyle{special_node3}=[none,scale=.1, fill=black]
\tikzstyle{semi-active}=[decoration={markings, mark=at position .5 with {\arrow{to}}},postaction={decorate}, ultra thick]


\begin{tikzpicture}[scale=0.5]
	\begin{pgfonlayer}{nodelayer}
		\node [style={leaf_node}, label={left:$a_2$}] (419) at (10, -8) {};
		\node [style={medium_label}] (420) at (17.25, -4.25) {$\N|_{\Y \textcolor{gray}{\cup \{x\}}}$};
		\node [style=none] (421) at (17, -6.25) {};
		\node [style=none] (422) at (17, -10.25) {};
		\node [style={small_label}] (423) at (17.5, -9.5) {$B$};
		\node [style={small_label}] (424) at (16, -9.5) {$A \textcolor{gray}{\cup \{x\}}$};
		\node [style={small_label}] (425) at (22, -7.5) {$\B$};
		\node [style={internal_node}] (426) at (14, -8) {};
		\node [style={internal_node}] (427) at (11, -8) {};
		\node [style={internal_node}] (428) at (11.75, -6.75) {};
		\node [style={internal_node}] (429) at (13.25, -9.25) {};
		\node [style={internal_node}] (430) at (11.75, -9.25) {};
		\node [style={special_node3}] (431) at (13.25, -6.75) {};
		\node [style={leaf_node}, label={above:$a_1$}] (432) at (11, -5.75) {};
		\node [style={leaf_node}, label={below:$a_3$}] (433) at (11, -10.25) {};
		\node [style={leaf_node}, label={below:$a_4$}] (434) at (14, -10.25) {};
		\node [style={leaf_node}, label={above:$b_2$}] (435) at (23.75, -5.75) {};
		\node [style={leaf_node}, label={below:$b_5$}] (436) at (21.5, -11) {};
		\node [style={leaf_node}, label={below:$b_6$}] (437) at (19.25, -10.25) {};
		\node [style={leaf_node}, label={right:$b_3$}] (438) at (27, -6.5) {};
		\node [style={leaf_node}, label={right:$b_4$}] (439) at (27, -9.5) {};
		\node [style={internal_node}] (440) at (25, -8) {};
		\node [style={internal_node}] (441) at (23.25, -8) {};
		\node [style={internal_node}] (442) at (19.75, -8) {};
		\node [style={internal_node}] (443) at (20.25, -6.75) {};
		\node [style={internal_node}] (444) at (22.75, -9.25) {};
		\node [style={internal_node}] (445) at (20.25, -9.25) {};
		\node [style={internal_node}] (446) at (22.75, -6.75) {};
		\node [style={internal_node}] (447) at (26.25, -7.25) {};
		\node [style={internal_node}] (448) at (26.25, -8.75) {};
		\node [style={special_node2}, label={above:$\textcolor{gray}{x}$}] (449) at (14, -5.75) {};
		\node [style={leaf_node}, label={below:$b_1$}] (450) at (31.25, -11.25) {};
		\node [style={leaf_node}, label={below:$b_3$}] (451) at (34.25, -11.25) {};
		\node [style={internal_node}] (452) at (32.75, -9.25) {};
		\node [style={internal_node}] (453) at (32, -10.25) {};
		\node [style={internal_node}] (454) at (33.5, -10.25) {};
		\node [style={internal_node}] (455) at (32.75, -7.5) {};
		\node [style={leaf_node}, label={above:$a_1$}] (456) at (31.25, -6) {};
		\node [style={medium_label}] (457) at (32.5, -4.25) {$\N|_{\{a_1, x, b_1, b_3\}}$};
		\node [style={leaf_node}, label={above:$x$}] (458) at (34.25, -6) {};
		\node [style={leaf_node}, label={below:$b_1$}] (459) at (36.75, -11.25) {};
		\node [style={leaf_node}, label={below:$b_6$}] (460) at (39.75, -11.25) {};
		\node [style={internal_node}] (461) at (38.25, -9.25) {};
		\node [style={internal_node}] (462) at (37.5, -10.25) {};
		\node [style={internal_node}] (463) at (39, -10.25) {};
		\node [style={internal_node}] (464) at (38.25, -7.5) {};
		\node [style={leaf_node}, label={above:$a_1$}] (465) at (36.75, -6) {};
		\node [style={medium_label}] (466) at (38, -4.25) {$\N|_{\{a_1, x, b_1, b_6\}}$};
		\node [style={leaf_node}, label={above:$x$}] (467) at (39.75, -6) {};
		\node [style={internal_node}] (468) at (21.5, -6.25) {};
		\node [style={internal_node}] (469) at (21.5, -9.75) {};
		\node [style={leaf_node}, label={above:$b_1$}] (470) at (21.5, -5) {};
		\node [style={small_label}] (471) at (14.25, -8.25) {$u$};
		\node [style={small_label}] (472) at (19.5, -8.25) {$v$};
	\end{pgfonlayer}
	\begin{pgfonlayer}{edgelayer}
		\draw [style={special_edge7}] (421.center) to (422.center);
		\draw [style=edge, bend right=15] (428) to (427);
		\draw [style=edge, bend left=15] (431) to (426);
		\draw [style=edge, bend left=15] (426) to (429);
		\draw [style={ret_arc}, bend right=15] (427) to (430);
		\draw [style={ret_arc}, bend left=15] (429) to (430);
		\draw [style=edge] (428) to (432);
		\draw [style=edge] (427) to (419);
		\draw [style=edge] (430) to (433);
		\draw [style=edge] (429) to (434);
		\draw [style=edge, bend right=15] (443) to (442);
		\draw [style={ret_arc}, bend right=15] (442) to (445);
		\draw [style={ret_arc}, bend left=15] (446) to (441);
		\draw [style={ret_arc}, bend right=15] (444) to (441);
		\draw [style=edge] (442) to (426);
		\draw [style={ret_arc}] (440) to (448);
		\draw [style={ret_arc}] (447) to (448);
		\draw [style=edge] (441) to (440);
		\draw [style=edge] (440) to (447);
		\draw [style=edge] (447) to (438);
		\draw [style=edge] (448) to (439);
		\draw [style=edge] (445) to (437);
		\draw [style=edge] (446) to (435);
		\draw [style={special_edge5}] (449) to (431);
		\draw [style={ret_arc}] (452) to (454);
		\draw [style={ret_arc}] (453) to (454);
		\draw [style=edge] (452) to (453);
		\draw [style=edge] (453) to (450);
		\draw [style=edge] (454) to (451);
		\draw [style=edge] (452) to (455);
		\draw [style=edge] (455) to (456);
		\draw [style=edge] (455) to (458);
		\draw [style={ret_arc}] (461) to (463);
		\draw [style={ret_arc}] (462) to (463);
		\draw [style=edge] (461) to (462);
		\draw [style=edge] (462) to (459);
		\draw [style=edge] (463) to (460);
		\draw [style=edge] (461) to (464);
		\draw [style=edge] (464) to (465);
		\draw [style=edge] (464) to (467);
		\draw [style=edge, bend right=15] (469) to (444);
		\draw [style=edge, bend left=15] (468) to (446);
		\draw [style=edge, bend left=345] (468) to (443);
		\draw [style={ret_arc}, bend left=15] (469) to (445);
		\draw [style=edge] (469) to (436);
		\draw [style=edge] (470) to (468);
		\draw [style=edge] (443) to (444);
		\draw [style=edge, bend left=15] (428) to (431);
	\end{pgfonlayer}
\end{tikzpicture}
\caption{\emph{Left:} The subnetwork of a semi-directed level-2 network $\N$ induced by the leaves $\Y = \{a_1, \ldots, a_4, b_1, \ldots , b_6\}$ (excluding the grey edge and the grey leaf $x$) and induced by the leaves $\Y \cup \{x\}$ (including the grey edge and the grey leaf $x$). The cut-edge $uv$ of $\N|_\Y$ induces the split $A | B$ (with $A = \{a_1, \ldots, a_4 \}$ and $B = \{b_1, \ldots, b_6 \}$) and $v$ is part of the blob $\B$. \emph{Right:} Two quarnets of $\N|_{\Y \cup \{x\}}$ with splits $a_1 x | b_1 b_3$ and $a_1 x | b_1 b_6$, respectively. By \cref{thm:determine_split} and since $\{a_1, b_1\} \cup B'$ with $B' = \{b_3, b_6\}$ is a distinguishing set of $\B$ in $\N|_\Y$, the splits of these two quarnets are enough to determine that $A \cup \{x\} |B$ is a split of $\N|_{\Y \cup \{x\}}$, without knowing anything else about $x$.}
\label{fig:split_theorem}
\end{figure}

We are now ready to prove the theorem lying at the heart of the algorithms in this paper. In particular, it suggests which quarnet-splits are important when we want to add a new leaf to the tree-of-blobs of a subnetwork (see again \cref{fig:split_theorem} for an illustration of the theorem). Note that we can always set $B'$ equal to $B \setminus \{b_1\}$ in this theorem if we have no further information on the reticulations of the network. Furthermore, in a semi-directed level-$\ell$ network, there is always a distinguishing set of size at most $\max\{ 3, \ell\}$ for any blob. This will be crucial for our level-1 reconstruction algorithm in \cref{sec:level-1}. Recall that in the next theorem, by definition, the vertex $v$ is on the side $B$ of the split $A|B$.

\begin{theorem}\label{thm:determine_split}
Given a semi-directed network $\N$ on $\X$, let $\Y \subset \X$ and $x \in \X \setminus \Y$. Let $A|B$ be a split in $ \N|_{\Y}$ induced by the cut-edge $uv$ such that $|B| \geq 2$ and $v$ is in a blob $\B$ of $\N|_{\Y}$. Let $a_1 \in A, b_1 \in B$ be arbitrary and let $B' \subseteq B \setminus \{b_1\}$ be such that $\{a_1, b_1\} \cup B'$ is a distinguishing set of $\B$, then the following are equivalent
\begin{enumerate}[label={(\roman*)}, noitemsep,topsep=0pt]
\item $A \cup \{x\} |B$ is a split in $\N|_{\Y \cup \{x\}}$;
\item $a_1 x | b_1 b_2$ is a quarnet-split of $\N|_{\Y \cup \{x\}}$ for all $b_2 \in B'$.
\end{enumerate}
\end{theorem}
\begin{proof}
That (i) implies (ii) follows readily from \cref{lem:splits} and \cref{obs:restriction}. For the other direction, let $A| B_1| \ldots | B_s$ be the partition induced by $\B$ (with $a_1 \in A$ and $b_1 \in B_1$). Since $\N$ has no 2-blobs, we have that $s\geq 2$ and thus a distinguishing set of $\B$ exists. Let $B' = \{b_2, \ldots, b_t\}$ be such that $D = \{a_1, b_1\} \cup B'$ contains as few leaves as possible, while still being a distinguishing set of $\B$. This implies that every $b_i \in B'$ is in a different $B_j$. Without loss of generality, we can then assume that $b_i \in B_i$ for all $1 \leq i \leq t \leq s$. Now suppose that $a_1 x | b_1 b_i$ is a split in $\N|_{\{a_1, b_1, b_i, x\}}$ for all $b_i \in B'$. We will prove that $a_1 x | b_1 b_*$ is also a quarnet-split of $\N$ for all $b_* \in B \setminus \{b_1, \ldots , b_t \}$. To this end, let $b_* \in B \setminus \{b_1, \ldots , b_t \}$ be arbitrary. We denote by $\B'$ the blob of $\N|_{\Y \cup \{x\}}$ corresponding to the (possibly smaller) blob $\B$ in $\N|_\Y$ and let $R_*$ be the (possibly empty) set of reticulations in the blob $\B'$ such that $b_*$ is below them in $\N|_{\Y \cup \{x\}}$. We first prove the following claim and then continue the proof of the theorem by considering four cases.

\emph{Claim:} There is some $d_* \in D \cup \{x \}$ with the property that $d_*$ is below all the reticulations of $R_*$ in $\N|_{\Y \cup \{x\}}$. \emph{Proof of claim:} If $|R_*|=0$, the claim is trivial, so we assume that $|R_*| > 0$. As with leaves, we say that a reticulation $r_2$ is \emph{below} another reticulation $r_1$ if there exists a partially directed path from $r_1$ to $r_2$. It suffices to show that there is some $r\in R_*$ below all reticulations in $R_* \setminus \{r\}$, since we can then simply let $d_* \in D \cup \{x\}$ be the leaf below $r$ (which exists because $D$ is a distinguishing set of $\B$, so $D \cup \{x\}$ is a distinguishing set of $\B'$). Towards a contradiction assume that there is no reticulation $r$ below all other reticulations in $R_* \setminus \{r\}$. Then, there are at least two reticulations $r_1$ and $r_2$ with no reticulation from $R_*$ below them. But for $b_*$ to be below both $r_1$ and $r_2$, a third reticulation has to be below both $r_1$ and $r_2$: a contradiction. $\blacksquare$


\emph{Case 1: $b_* \in B_1 \setminus \{b_1\}$.} By our assumption, we know that $a_1 x | b_1 b_2$ is a split in $\N|_{\{a_1, b_1, b_2, x\}}$. But since $b_1$ and $b_*$ are both part of the set $B_1$, we have that $a_1 x | b_2 b_*$ is then also a split in $\N|_{\{a_1, b_2, b_*, x\}}$. Using \cref{lem:splits}, these two quarnet-splits show that $a_1 x | b_1 b_2 b_*$ is a split in $\N|_{\{a_1, b_1, b_2, b_*, x\}}$, which means that $a_1 x | b_1 b_*$ is a split in $\N|_{\{a_1, b_1, b_* x\}}$.

\emph{Case 2: $b_* \in B_i \setminus \{b_i\}$ with $2 \leq i \leq t$.} By our assumption, we know that $a_1 x | b_1 b_i$ is a split in $\N|_{\{a_1, b_1, b_i, x\}}$. Clearly, since $\{b_*, b_i \} \subseteq B_i$, the quarnet $\N|_{\{a_1, b_1, b_*, x\}}$ has $a_1 x | b_1 b_*$ as its split.

\emph{Case 3: $b_* \in B_i$ with $2 \leq t < i \leq s$ and $d_* \in B' \subseteq D \cup \{x\}$}. Let $b' = d_*$ and recall that by our assumption, $a_1 x | b_1 b'$ is a split in $\N|_{\{a_1, b_1, b', x\}}$. By definition, this means that there is some non-empty set of edges $E$ in $\N|_{\Y \cup \{x\}}$ that are part of every up-down path between leaves from $\{a_1, x\}$ and $\{b_1, b'\}$. Towards a contradiction, suppose that $a_1 x | b_1 b_*$ is not a split in $\N|_{\{a_1, b_1, b_*, x\}}$, which by \cref{lem:splits} means that $a_1 x | b_1 b' b_*$ is also not a split in $\N|_{\{a_1, b_1, b', b_*, x\}}$. Thus, the addition of $b_*$ to $\N|_{\{a_1, b_1, b', x\}}$ allows for bypassing the edges in $E$ when using up-down paths to connect the leaves in $\{a_1, x\}$ with those in $\{b_1, b'\}$. The only new ways to connect leaves $\{a_1, x\}$ with $\{b_1, b'\}$ when taking the subnetwork induced by $\{a_1, b_1, b', b_*, x\}$ (instead of $\{a_1, b_1, b', x\}$) is that in $\N|_{\{a_1, b_1, b', b_*, x\}}$ we can first take an up-down path from $\{a_1, x\}$ to $b_*$ and then one from $b_*$ to $\{b_1, b'\}$, whereas in $\N|_{\{a_1, b_1, b', x\}}$ this bypass is not possible and we can only take up-down paths directly from $\{a_1, x\}$ to $\{b_1, b'\}$. Now note that up-down paths were also characterized by Xu and An{\'e} \cite{xu2023identifiability} as paths without \emph{v-structures}: path segments containing both the reticulation edges corresponding to a reticulation vertex. From this characterization it becomes clear that the previously mentioned `bypass' of the edges in $E$ is possible in $\N|_{\{a_1, b_1, b', b_*, x\}}$ and not in $\N|_{\{a_1, b_1, b', x\}}$, because in $\N|_{\{a_1, b_1, b', x\}}$ the `bypass' consists of one path directly from $\{a_1, x\}$ to $\{b_1, b'\}$ with a v-structure on it (and we can not first go to $b_*$). Hence, $b_*$ is below the reticulation corresponding to the v-structure in $\N|_{\{a_1,x,b_1,b',b_*\}}$, while none of $\{a_1, x , b_1, b'\}$ is below this reticulation. It follows that $b_*$ is also below some reticulation $r \in R_*$ in $\N|_{\Y \cup \{x\}}$, while none of $\{a_1, x , b_1, b'\}$ is below $r$: a contradiction with the fact that $b' = d_*$ is below all reticulations in $R_*$ (see the claim). Therefore, we have that $a_1 x |b_1 b_*$ is a split in $\N|_{\{a_1, x, b_1, b_* \}}$.

\emph{Case 4: $b_* \in B_i$ with $2 \leq t < i \leq s$ and $d_* \in \{a_1, b_1, x\} \subseteq D \cup \{x\}$.} In this case, pick $b' = b_2 \in D$. We again have that $a_1 x |b_1 b'$ is a split in $\N|_{\{a_1, b_1, b', x\}}$. As in the previous case, assume towards a contradiction that $a_1 x | b_1 b_*$ is not a split in $\N|_{\{a_1, b_1, b_*, x\}}$. The exact same reasoning then shows that $b_*$ is below some some reticulation $r$ of $R_*$ in $\N|_{\Y \cup \{x\}}$, while none of $\{a_1, x , b_1, b'\}$ is below $r$. This is again a contradiction because $d_* \in \{a_1, b_1, x\}$ is below all reticulations in $R_*$ (see the claim). Therefore, we again have that $a_1 x |b_1 b_*$ is a split in $\N|_{\{a_1, x, b_1, b_* \}}$.

All in all, the four cases show that $a_1 x | b_1 b_*$ is a split in $\N|_{\{a_1, b_1, b_*, x\}}$ for all $b_* \in B \setminus \{b_1\}$. Since $A|B$ is a split in $\N|_\Y$, \cref{lem:splits} tells us that $a_1 a_2 |b_1 b_*$ is a split in $\N|_{\{a_1, a_2, b_1, b_*\}}$ for all $a_2 \in A \setminus \{a_1\}$, $b_* \in B \setminus \{b_1\}$. Together, we thus get that $a_1 a_* |b_1 b_*$ is a split in $\N|_{\{a_1, a_*, b_1, b_*\}}$ for all $a_* \in (A \cup \{x\})\setminus \{a_1\}$ and $b_* \in B \setminus \{b_1\}$. So, again by \cref{lem:splits}, $A\cup \{x\}|B$ is a split in $\N|_{\Y \cup \{x \}}$. This proves the backward direction for any $B'$ such that $D$ has as few leaves as possible. It follows trivially that it also holds for any other allowed set $B'$.
\end{proof}

\section{Reconstructing a tree-of-blobs from quarnet-splits}\label{sec:blobtree}

In this section we will describe an algorithm that reconstructs the tree-of-blobs of a semi-directed network $\N$ on $\X$, using only its quarnet-splits. Our approach to achieve this will be to `grow' our tree-of-blobs one leaf at a time. In particular, we will maintain the property that every intermediate tree on leaf set $\Y \subseteq \X$ will be the tree-of-blobs $\T (\N|_\Y )$. As mentioned in \cref{sec:preliminaries}, this tree can be more refined than $\T (\N )$. Thus, we might need to contract parts of the tree before attaching a new leaf. Where to attach this leaf and what parts to contract will be based on the splits of the tree-of-blobs, which can be determined with \cref{thm:determine_split} and the quarnet-splits. To this end, we need the following concepts.

\begin{definition}\label{def:active_edge}
Given a semi-directed network $\N$ on $\X$ and $\Y \subset \X$, let $uv$ be an edge in $\T ( \N|_{\Y})$ inducing the split $A|B$ and $x \in \X \setminus \Y$. The edge $uv$ is one of the following
\begin{itemize}[noitemsep,topsep=0pt]
    \item a \emph{strong stem edge (for $x$)}, if $A |B \cup \{x\}$ and $A  \cup \{x\} |B$ are splits in $\N|_{\Y \cup \{x\}}$;
    \item a \emph{weak stem edge (for $x$)}, if neither $A |B \cup \{x\}$ nor $A \cup \{x\} |B$ are splits in $\N|_{\Y \cup \{x\}}$;
    \item a \emph{pointing edge (for $x$) with orientation $uv$}, if $A |B \cup \{x\}$ is a split in $\N|_{\Y \cup \{x\}}$ but $A  \cup \{x\} |B$ is not;
    \item a \emph{pointing edge (for $x$) with orientation $vu$}, if $A |B \cup \{x\}$ is not a split in $\N|_{\Y \cup \{x\}}$ but $A  \cup \{x\} |B$ is.
\end{itemize}
An internal vertex $v$ of $\T ( \N|_{\Y})$ is a \emph{stem vertex (for $x$)} if all its incident edges $uv$ in $\T ( \N|_{\Y})$ are pointing edges for $x$ with orientation $uv$. A subtree of $\T ( \N|_\Y )$ is a \emph{weak stem subtree (for $x$)} if it is a maximal subtree containing only weak stem edges for $x$.    
\end{definition}

If the leaf $x$ is clear from the context, we often omit `for $x$' from these descriptions. We use the terms `pointing' and `orientation' to emphasize that the pointing edges are not actual directed edges in the tree-of-blobs. To illustrate these different types of edges and vertices we refer to \cref{fig:treestem}. One can check that the edges and vertices in the depicted trees-of-blobs exactly follow the previous definition. What stands out from the figure is that the highlighted edges and vertex seem to be the focal points when obtaining $\T (\N|_{\X_{i+1}})$ from $\T( \N|_{\X_i})$. This motivates the next definition of a \emph{stem} (not to be confused with its unrelated namesake from \cite{huebler2019constructing}).

\begin{definition}[Stem]
Given a semi-directed network $\N$ on $\X$ and $\Y \subset \X$ with $x \in \X \setminus \Y$, a \emph{stem (for $x$)} of $\T( \N |_\Y)$ is a strong stem edge for $x$, stem vertex for $x$, or weak stem subtree for $x$.
\end{definition}

\begin{figure}[htb]
\centering
\begin{tikzpicture}[scale=0.45]
	\begin{pgfonlayer}{nodelayer}
		\node [style={main_label}] (26) at (1.25, 8) {$\N$};
		\node [style={internal_node}] (27) at (3.5, 5.75) {};
		\node [style={internal_node}] (28) at (3.5, 0.75) {};
		\node [style={internal_node}] (31) at (2, 5.25) {};
		\node [style={internal_node}] (32) at (1, 4) {};
		\node [style={internal_node}] (33) at (1, 2.5) {};
		\node [style={internal_node}] (34) at (2, 1.25) {};
		\node [style={internal_node}] (35) at (5, 1.25) {};
		\node [style={internal_node}] (36) at (6, 2.5) {};
		\node [style={internal_node}] (37) at (6, 4) {};
		\node [style={internal_node}] (38) at (5, 5.25) {};
		\node [style={leaf_node}, label={left:1}] (39) at (1.25, 6) {};
		\node [style={leaf_node}, label={left:2}] (40) at (0, 1.75) {};
		\node [style={leaf_node}, label={below:3}] (41) at (1.5, 0.25) {};
		\node [style={leaf_node}, label={below:4}] (42) at (5.75, 0.25) {};
		\node [style={leaf_node}, label={right:5}] (43) at (7, 1.75) {};
		\node [style={leaf_node}, label={above:9}] (44) at (5.75, 6.25) {};
		\node [style={leaf_node}, label={above:10}] (45) at (3.5, 6.75) {};
		\node [style={internal_node}] (46) at (7.5, 4.75) {};
		\node [style={internal_node}] (47) at (8, 5.75) {};
		\node [style={internal_node}] (48) at (8.75, 4.75) {};
		\node [style={internal_node}] (49) at (8.25, 6.75) {};
		\node [style={internal_node}] (50) at (8, 5.75) {};
		\node [style={leaf_node}, label={right:6}] (51) at (9.75, 4.25) {};
		\node [style={leaf_node}, label={above:8}] (52) at (7.75, 7.75) {};
		\node [style={leaf_node}, label={above:7}] (53) at (9.25, 7.5) {};
		\node [style={internal_node}] (54) at (15, 7.25) {};
		\node [style={internal_node}] (55) at (16.5, 7.25) {};
		\node [style={internal_node}] (56) at (18, 7.25) {};
		\node [style={internal_node}] (57) at (19.5, 7.25) {};
		\node [style={leaf_node}, label={above:1}] (58) at (15, 8.5) {};
		\node [style={leaf_node}, label={left:2}] (59) at (13.75, 7.25) {};
		\node [style={leaf_node}, label={below:3}] (60) at (15, 6) {};
		\node [style={leaf_node}, label={below:4}] (61) at (16.5, 6) {};
		\node [style={leaf_node}, label={below:5}] (62) at (18, 6) {};
		\node [style={leaf_node}, label={right:6}] (63) at (20.5, 6) {};
		\node [style={leaf_node}, label={right:7}] (64) at (20.5, 8.25) {};
		\node [style={medium_label}] (65) at (18, 9.5) {$\T (\N|_{\X_7})$};
		\node [style={internal_node}] (66) at (24.75, 7.25) {};
		\node [style={internal_node}] (67) at (26.25, 7.25) {};
		\node [style={internal_node}] (68) at (27.75, 7.25) {};
		\node [style={internal_node}] (69) at (29.25, 7.25) {};
		\node [style={leaf_node}, label={above:1}] (70) at (24.75, 8.5) {};
		\node [style={leaf_node}, label={left:2}] (71) at (23.5, 7.25) {};
		\node [style={leaf_node}, label={below:3}] (72) at (24.75, 6) {};
		\node [style={leaf_node}, label={below:4}] (73) at (26.25, 6) {};
		\node [style={leaf_node}, label={below:5}] (74) at (27.75, 6) {};
		\node [style={leaf_node}, label={below:6}] (75) at (29.25, 6) {};
		\node [style={medium_label}] (77) at (28.25, 9) {$\T (\N|_{\X_8})$};
		\node [style={internal_node}] (78) at (30.75, 7.25) {};
		\node [style={leaf_node}, label={right:7}] (79) at (31.75, 6) {};
		\node [style={leaf_node}, label={right:8}] (80) at (31.75, 8.5) {};
		\node [style={internal_node}] (84) at (17.75, 1.25) {};
		\node [style={leaf_node}, label={above:1}] (85) at (15.25, 2.5) {};
		\node [style={leaf_node}, label={above:2}] (86) at (14.25, 2.25) {};
		\node [style={leaf_node}, label={below:3}] (87) at (14.25, 0.25) {};
		\node [style={leaf_node}, label={below:4}] (88) at (15.25, 0) {};
		\node [style={leaf_node}, label={below:5}] (89) at (16.25, 0.25) {};
		\node [style={leaf_node}, label={below:6}] (90) at (17.75, 0) {};
		\node [style={medium_label}] (91) at (18.25, 3.75) {$\T (\N|_{\X_9})$};
		\node [style={internal_node}] (92) at (19.25, 1.25) {};
		\node [style={leaf_node}, label={right:7}] (93) at (20.25, 0) {};
		\node [style={leaf_node}, label={right:8}] (94) at (20.25, 2.5) {};
		\node [style={leaf_node}, label={above:9}] (95) at (16.25, 2.25) {};
		\node [style={active_vertex}] (96) at (15.25, 1.25) {};
		\node [style={internal_node}] (97) at (28.5, 1.25) {};
		\node [style={leaf_node}, label={above:1}] (98) at (25, 2.25) {};
		\node [style={leaf_node}, label={left:2}] (99) at (24.75, 1.25) {};
		\node [style={leaf_node}, label={below:3}] (100) at (25, 0.25) {};
		\node [style={leaf_node}, label={below:4}] (101) at (26, 0) {};
		\node [style={leaf_node}, label={below:5}] (102) at (27, 0.25) {};
		\node [style={leaf_node}, label={below:6}] (103) at (28.5, 0) {};
		\node [style={medium_label}] (104) at (29, 3.25) {$\T (\N)$};
		\node [style={internal_node}] (105) at (30, 1.25) {};
		\node [style={leaf_node}, label={right:7}] (106) at (31, 0) {};
		\node [style={leaf_node}, label={right:8}] (107) at (31, 2.5) {};
		\node [style={leaf_node}, label={above:9}] (108) at (27, 2.25) {};
		\node [style={leaf_node}, label={above:10}] (110) at (26, 2.5) {};
		\node [style={internal_node}] (111) at (26, 1.25) {};
	\end{pgfonlayer}
	\begin{pgfonlayer}{edgelayer}
		\draw [style={ret_arc}, bend left, looseness=1.25] (32) to (28);
		\draw [style={ret_arc}, bend right=15] (34) to (28);
		\draw [style={ret_arc}, bend left=15] (27) to (38);
		\draw [style={ret_arc}, bend right=15] (37) to (38);
		\draw [style=edge, bend left=15] (37) to (36);
		\draw [style=edge, bend left=15] (36) to (35);
		\draw [style=edge, bend left=15] (35) to (28);
		\draw [style=edge, bend left=15] (34) to (33);
		\draw [style=edge, bend left=15] (33) to (32);
		\draw [style=edge, bend left=15] (32) to (31);
		\draw [style=edge, bend left=15] (31) to (27);
		\draw [style=edge] (39) to (31);
		\draw [style=edge] (45) to (27);
		\draw [style=edge] (38) to (44);
		\draw [style=edge] (40) to (33);
		\draw [style=edge] (34) to (41);
		\draw [style=edge] (42) to (35);
		\draw [style=edge] (36) to (43);
		\draw [style=edge] (37) to (46);
		\draw [style=edge] (50) to (46);
		\draw [style=edge] (50) to (49);
		\draw [style={ret_arc}] (50) to (48);
		\draw [style={ret_arc}] (46) to (48);
		\draw [style=edge] (52) to (49);
		\draw [style=edge] (49) to (53);
		\draw [style=edge] (48) to (51);
		\draw [style=active] (64) to (57);
		\draw [style=semi-active] (58) to (54);
		\draw [style=semi-active] (59) to (54);
		\draw [style=semi-active] (60) to (54);
		\draw [style=semi-active] (54) to (55);
		\draw [style=semi-active] (55) to (56);
		\draw [style=semi-active] (56) to (57);
		\draw [style=semi-active] (63) to (57);
		\draw [style=semi-active] (62) to (56);
		\draw [style=semi-active] (61) to (55);
		\draw [style=semi-active] (70) to (66);
		\draw [style=semi-active] (71) to (66);
		\draw [style=semi-active] (72) to (66);
		\draw [style=semi-active] (74) to (68);
		\draw [style=semi-active] (73) to (67);
		\draw [style=passive] (66) to (67);
		\draw [style=passive] (67) to (68);
		\draw [style=semi-active] (69) to (68);
		\draw [style=semi-active] (75) to (69);
		\draw [style=semi-active] (78) to (69);
		\draw [style=semi-active] (80) to (78);
		\draw [style=semi-active] (79) to (78);
		\draw [style=semi-active] (90) to (84);
		\draw [style=semi-active] (92) to (84);
		\draw [style=semi-active] (94) to (92);
		\draw [style=semi-active] (93) to (92);
		\draw [style=semi-active] (85) to (96);
		\draw [style=semi-active] (95) to (96);
		\draw [style=semi-active] (86) to (96);
		\draw [style=semi-active] (87) to (96);
		\draw [style=semi-active] (88) to (96);
		\draw [style=semi-active] (89) to (96);
		\draw [style=semi-active] (84) to (96);
		\draw [style=edge] (107) to (105);
		\draw [style=edge] (105) to (106);
		\draw [style=edge] (105) to (97);
		\draw [style=edge] (97) to (103);
		\draw [style=edge] (97) to (111);
		\draw [style=edge] (111) to (101);
		\draw [style=edge] (111) to (102);
		\draw [style=edge] (111) to (100);
		\draw [style=edge] (111) to (99);
		\draw [style=edge] (111) to (98);
		\draw [style=edge] (111) to (110);
		\draw [style=edge] (111) to (108);
	\end{pgfonlayer}
\end{tikzpicture}
\caption{A semi-directed level-2 network $\N$ on $\X = \{ 1, \ldots, 10 \}$. On the right we have the tree-of-blobs $\T (\N)$ of the whole network and three trees-of-blobs $\T (\N|_{\X_{i}} )$, with $\X_i$ the set of the first $i$ leaves of $\X$. In each of the three intermediate trees-of-blobs, the oriented edges indicate pointing edges. The stems of the trees-of-blobs are highlighted and are in order: a strong stem edge in thick blue, a weak stem subtree (in this case a path) in thick red, and a stem vertex in thick blue.}
\label{fig:treestem}
\end{figure}

\cref{fig:treestem} suggests that there is a unique stem and that the pointing edges are oriented towards this stem. We prove this in the next lemma and therefore we thus refer to \emph{the} stem from now on. Furthermore, as the name suggests, the unique stem is the location where we attach the new leaf $x$ to the tree-of-blobs. Exactly how this needs to be done is also outlined in the next lemma. It is enough to consider the following two attachments for a given leaf $x$. By \emph{attaching $x$ to an edge $uv$}, we mean subdividing $uv$ with a vertex $y$ to create two new edges $uy$ and $yv$, and then adding a new edge $yx$. Similarly, \emph{attaching $x$ to a vertex $v$} means adding the edge~$vx$. 
\begin{lemma}\label{lem:unique_tree_stem}
Let $\N$ be a semi-directed network on $\X$ and let $\Y \subset \X$. Given the tree-of-blobs~$\T ( \N|_{\Y})$ and some $x \in \X \setminus \Y$, there is a unique stem and the orientation of any pointing edge in $\T ( \N|_{\Y})$ is towards this stem. Furthermore,
\begin{enumerate}[label={(\alph*)},noitemsep,topsep=0pt]
    \item if the stem is a strong stem edge, then the tree-of-blobs $\T ( \N|_{\Y \cup \{ x\}})$ can be obtained from $\T ( \N|_{\Y})$ by attaching $x$ to this edge;
    \item if the stem is a weak stem subtree, then the tree-of-blobs $\T ( \N|_{\Y \cup \{ x\}})$ can be obtained from $\T ( \N|_{\Y})$ by contracting this subtree to a single vertex and attaching $x$ to it;
    \item if the stem is a stem vertex, the tree-of-blobs $\T ( \N|_{\Y \cup \{ x\}})$ can be obtained from $\T ( \N|_{\Y})$ by attaching $x$ to this vertex.
\end{enumerate}
\end{lemma}
\begin{proof}
We let the blob $\B$ of $\N|_{\Y \cup \{ x\}}$ be the blob adjacent to $x$, i.e. the blob with whom $x$ shares an incident edge. $\N|_{\Y \cup \{ x\}}$ contains no 2-blobs, so when taking the subnetwork of $\N|_{\Y \cup \{ x\}}$ induced by $\Y$ we can just delete the leaf $x$. Then, we only need to do some operations on the blob $\B$. The other parts of the network remain the same. We distinguish three mutually exclusive cases covering what happens with $\B$.

(a): $\B$ completely disappears. This means that the deletion of $x$ made $\B$ into a 2-blob, which was then subsequently suppressed. Then, there is a cut-edge $e$ in $\N|_\Y$ that is the result of this suppression. Since all other parts of the network remain the same, we can obtain $\T ( \N|_{\Y \cup \{ x\}})$ from $\T ( \N|_{\Y})$ by attaching $x$ to the unique edge $uv$ that induces the same split $A|B$ as the cut-edge $e$. Then, $A| B \cup \{x\}$ and $A \cup \{x\}|B$ are splits in $\T ( \N|_{\Y \cup \{ x\}})$, so $uv$ is a strong stem edge. Since the rest of the tree $\T ( \N|_{\Y \cup \{ x\}})$ remains the same as in
$\T ( \N|_{\Y})$, it is easy to check that all other edges in $\T ( \N|_{\Y})$ are pointing edges oriented towards $uv$.

(b): $\B$ splits up into multiple blobs $\B'_i$. Then, the blobs $\B_i'$ are connected to each other. So, the vertices in $\T ( \N|_{\Y})$ that represent these blobs form a subtree $T$. Again, since all other parts of the network remain the same, we can obtain $\T ( \N|_{\Y \cup \{ x\}})$ from $\T ( \N|_{\Y})$ by contracting the edges in $T$ and then attaching $x$ to the created vertex. Furthermore, if $uv$ is an edge in $T$ that induces the split $A|B$, then neither $A| B \cup \{x\}$ nor $A \cup \{x\}|B$ can be a split in $\T ( \N|_{\Y \cup \{ x\}})$. Thus, every such edge $uv$ is a weak stem edge and so $T$ is a weak stem subtree. By similar reasoning as before, all other edges in $\T ( \N|_{\Y})$ are pointing edges oriented towards $T$.

(c): $\B$ becomes a single blob $\B'$. Then, we can obtain $\T ( \N|_{\Y \cup \{ x\}})$ from $\T ( \N|_{\Y})$ by attaching $x$ to the vertex $v$ representing $\B'$. Again, since all other parts of the network (and thus of the tree-of-blobs) remain the same, all edges in $\T ( \N|_{\Y})$ are pointing edges oriented towards $v$, showing that $v$ is a stem vertex.
\end{proof}

With \cref{lem:unique_tree_stem} we now know how to attach a leaf to a stem (see \cref{fig:treestem} for examples). Thus, we have reduced the problem of reconstructing a tree-of-blobs from quarnet-splits to finding the stem in a tree-of-blobs. The following corollary of \cref{thm:determine_split} shows how this is done in quadratic time.
\begin{corollary}\label{cor:find_treestem}
Let $\N$ be a semi-directed network on $\X$ and let $\Y \subset \X$ contain $k$ leaves. Given the tree-of-blobs $\T ( \N|_\Y)$, some $x \in \X \setminus \Y$ and the quarnet-splits of $\N$, one can find the stem for $x$ of $\T ( \N|_\Y)$ in $\bigO(k^2)$ time using $\bigO(k^2)$ quarnet-splits of~$\N$.
\end{corollary}
\begin{proof}
We know from the previous lemma that there is a unique stem in $\T ( \N|_\Y )$. To find this stem, it is enough to check for each of the at most $\bigO(k)$ edges in the tree-of-blobs whether it is a strong stem edge, weak stem edge, or pointing edge. If such an edge induces a split $A|B$ with $|B| = 1$, this is trivial since then $A \cup \{ x\}|B$ is always a split. Otherwise, if $|B| \geq 2$, we can use \cref{thm:determine_split} to determine whether $A \cup \{ x\}|B$ is a split in $\bigO(k)$ time, using $\bigO(k)$ quarnet-splits of $\N$. In particular, we arbitrarily pick $a_1 \in A$, $b_1 \in B$ and then check if $a_1 x | b_1 b_2$ is a quarnet-split for all $b_2 \in B' = B \setminus \{b_1\}$. So, we can always determine whether $A \cup \{ x\}|B$ is a split. Hence, we know in $\bigO(k)$ time (using $\bigO(k)$ quarnet-splits) whether the edge is a strong stem edge, weak stem edge, or pointing edge, which proves the result.
\end{proof}

In \cref{alg:blobtree}, we can now present a concise formulation of how to reconstruct the tree-of-blobs of a semi-directed network. \cref{fig:treestem} again serves as an example of a few iterations of this algorithm. The time complexity of the algorithm is stated in \cref{thm:blobtree_alg}, whose proof follows directly from the two previous results. As mentioned in \cref{sec:preliminaries}, the quarnet-splits that are used as input for \cref{alg:blobtree} can also be obtained from the displayed quartets of a semi-directed network.

\begin{algorithm}[htb]
\caption{Reconstruction of the tree-of-blobs of a semi-directed network}\label{alg:blobtree}
\Input{quarnet-splits of a semi-directed network $\N$ on $\X = \{x_1, \ldots, x_n \}$}
\Output{tree-of-blobs of $\N$}
$\X_i \gets \{x_1, \ldots, x_i\}$ for all $i \in \{1, \ldots, n\}$\\

$\T ( \N|_{\X_2}) \gets$ single edge $\{x_1, x_2\}$\\

\For(\tcp*[f]{if $n=2$, $\{3, \ldots, n\} = \emptyset$}){$i \in \{3, \ldots, n \}$}{
    find the stem of $\T ( \N|_{\X_{i-1}})$, using \cref{cor:find_treestem} and the quarnet-splits of $\N$\\
    $\T ( \N|_{\X_i})$ is constructed from $\T ( \N|_{\X_{i-1}})$ as described in \cref{lem:unique_tree_stem}\\
}
\Return{$\T ( \N|_{\X_n})$}
\end{algorithm}

\begin{theorem}\label{thm:blobtree_alg}
Given the quarnet-splits of a semi-directed network $\N$ on $\X = \{ x_1, \ldots, x_n \}$, \cref{alg:blobtree} reconstructs the tree-of-blobs of $\N$ in $\bigO(n^3)$ time using $\bigO(n^3)$ quarnet-splits of $\N$.
\end{theorem}

We refer to \cref{alg:blobtree} as a \emph{tree-of-blobs reconstruction algorithm}: an algorithm that reconstructs the tree-of-blobs of a given semi-directed network, using only its quarnet-splits. Any tree-of-blobs reconstruction algorithm needs to use $\Omega (n)$ quarnet-splits since it needs to consider every leaf at least once. An interesting question that arises, is whether any tree-of-blobs reconstruction algorithm needs $\Omega (n^3)$ quarnet-splits, which would show that \cref{alg:blobtree} is optimal. Although we are not able to answer this question affirmatively, we do provide a worst-case lower bound of $\Omega (n \log n + \ell \cdot n)$ quarnet-splits for level-$\ell$ networks. This simplifies to $\Omega (n^2)$ quarnet-splits for networks of unbounded level (because then $\ell$ could be $\Omega (n)$, see e.g. \cref{fig:lower_bound}), narrowing the gap to $\Theta (n)$.

\begin{proposition}\label{prop:lower_bound}
Given the quarnet-splits of a semi-directed level-$\ell$ network $\N$ on $\X = \{x_1, \ldots, x_n\}$, any algorithm using only quarnet-splits to reconstruct the tree-of-blobs of $\N$ needs to use $\Omega (n \log n + \ell \cdot n)$ quarnet-splits.
\end{proposition}
\begin{proof}
The first part of the bound is based on reasoning from \cite{kannan1996determining}, who provide a lower bound on the number of triplets (3-leaf subtrees) that need to be used for tree reconstruction. They show that there are $2^{\Omega(n \log n)}$ possible binary phylogenetic trees on $n$ leaves. This is also a lower bound on the number of trees-of-blobs of all $n$-leaf semi-directed strict level-$\ell$ networks. Clearly, there are only four possible quarnet-splits for a given set of four leaves (including the possibility that a quarnet has no non-trivial split). Then, by an information-theoretic argument, any tree-of-blobs reconstruction algorithm must use $\log_{4} (2^{\Omega(n \log n)}) = \Omega (n \log n)$ quarnet-splits in the worst case. This shows the first part of the bound for any~$\ell$.

Consider the semi-directed strict level-$\ell$ network $\N$ with $p\geq 2$ internal blobs from \cref{fig:lower_bound}. For all $k \in \{1, \ldots, p \}$ and $i, j \in \{1, \ldots, \ell \}$ with $i\neq j$, let $\N_{i, j, k}$ be the second network from \cref{fig:lower_bound} which is also strict level-$\ell$ for $p \geq 2$. Towards the second bound, suppose that $\ell \geq 2$. For simplicity, we also assume that $n\geq 4$ is a multiple of $\ell$ such that $n = p \cdot \ell$. For values of $n$ which are not multiples of $\ell$ the following exposition can trivially be adapted by letting one of the blobs be bigger (but with a similar structure) to accommodate for the remainder of the leaves.
Every quarnet-split that does not contain both $x_{k, i}$ and $x_{k, j}$ is the exact same in $\N$ and $\N_{i,j, k}$. (Note that the quarnets might differ since $i$ is not necessarily $j-1$, but the quarnet-splits do not.) Since $\N$ has a different tree-of-blobs than the networks $\N_{i,j, k}$, every tree-of-blobs reconstruction algorithm must consider at least one quarnet-split for every combination of $k \in \{1, \ldots, p \}$ and $i \neq j \in \{1, \ldots, \ell \}$. Otherwise, such an algorithm can never unambiguously construct the tree-of-blobs of $\N$. This means that such an algorithm needs to consider at least $\frac12 \cdot p \cdot \ell \cdot (\ell-1) = \frac12 \cdot n \cdot (\ell-1)$ quarnet-splits of $\N$. The bound then immediately follows.
\end{proof}

\begin{figure}[htb]
\centering
\input{fig_lower_bound}
\caption{A semi-directed strict level-$\ell$ network $\N$ with $p$ internal blobs and a set of semi-directed strict level-$\ell$ networks $\N_{i, j, k}$ on the same leaf set with $p+1$ internal blobs, where $i,j \in \{1, \ldots, \ell\}$ such that $i \neq j$ and $k \in \{1, \ldots, p \}$. All networks have $n = \ell\cdot p$ leaves, where $\ell \geq 2$ and $p\geq 2$.}
\label{fig:lower_bound}
\end{figure}

\section{Reconstructing a level-1 network from quarnets}\label{sec:level-1}

Whereas the previous section was aimed at reconstructing the tree-of-blobs of a general semi-directed network, this section is solely devoted to semi-directed level-1 networks. With this restriction we aim for more than in the previous section: the reconstruction of the full network. We first show how to construct a canonical form of such a network from $\bigO (n \log n)$ of its quarnet-splits, after which we use the full quarnets to reconstruct the whole network.

\subsection{Canonical level-1 networks}
Recall from \cref{sec:preliminaries} that every internal blob in a semi-directed level-1 network is either a single vertex or a $k$-cycle. We say that such a cycle is a \emph{large cycle} if it is neither a triangle nor a 4-cycle. An edge between two adjacent non-reticulation vertices $u$ and $v$ of a cycle in a (possibly non-binary) semi-directed level-1 network is a \emph{cycle contraction edge} if $u$ or $v$ is adjacent to the reticulation vertex of the cycle and if $u$ and $v$ are both degree-3 vertices. Clearly, every 4-cycle and every large cycle of a binary semi-directed level-1 network has two distinct cycle contraction edges: one on each side of the reticulation. We are now ready to define the canonical form that is of interest in this section and we refer to the top part of \cref{fig:canonical_network} for an example.

\begin{definition}[Canonical form]\label{def:canonical_network}
Given a semi-directed level-1 network $\N$ on $\X$, the \emph{canonical form} $\C$ of $\N$ is obtained by contracting every triangle and 4-cycle to a single vertex; and by contracting the cycle contraction edges of every large cycle.
\end{definition}

\begin{figure}[htb]
\centering
\input{fig_canonical}
\caption{\emph{Top:} A semi-directed level-1 network $\N$ on $\X = \{1, \ldots, 14\}$ and its canonical form $\C$ obtained by contracting all dotted edges. \emph{Bottom:} The subnetwork of $\N$ induced by $\Y= \{2, 3, 4, 5, 6, 11, 12 \}$, the subnetwork of $\C$ induced by $\Y$, and the canonical subnetwork of $\C$ induced by $\Y$ obtained from $\C|_\Y$ by contracting all dotted edges, which either form a 4-blob, or they are a cycle contraction edge between two degree-3 vertices. Note that this last canonical network is also the canonical form of $\N|_\Y$.}
\label{fig:canonical_network}
\end{figure}

Sometimes we say that $\C$ is a \emph{canonical (level-1) network} when we mean that it is the canonical form of a semi-directed level-1 network $\N$. Note that a canonical network is also a semi-directed level-1 network, albeit with a maximum degree of four, instead of three. Hence, we straightforwardly extend some definitions, such as blobs and trees-of-blobs, in a trivial way to canonical networks. It is easy to see that the canonical form of a semi-directed level-1 network can be viewed as a structure that is more refined than the tree-of-blobs, yet less refined than the original network. Moreover, the tree-of-blobs of the canonical form is equal to the tree-of-blobs of the original network. It will follow from \cref{thm:alg_canonical_network} that, as with the tree-of-blobs, the canonical form $\C$ of a network $\N$ can be reconstructed from the quarnet-splits of $\N$. As we will see in \cref{subsec:canonical_lev1}, it is even the most refined network one can construct from quarnet-splits. 

In this section, we are mostly concerned with the canonical form of a subnetwork, i.e. $ (\N|_\Y)^c$ for some $\Y \subseteq \X$. Unfortunately, we might have that $(\N|_\Y)^c \neq \C|_\Y $ (consider for example the set $\Y$ and the two networks $(\N|_\Y)^c$ and $\C|_\Y $ in \cref{fig:canonical_network}). Thus, it matters whether we first induce a subnetwork, or first create the canonical form. Hence, we define the following procedure for canonical networks which does have the desired property (see \cref{fig:canonical_network} for an example). Note that a 3-blob (resp. 4-blob) in a canonical network need not necessarily be a triangle (resp. 4-cycle) due to the possible degree-4 vertices.

\begin{definition}[Canonical subnetwork]\label{def:canonical_restriction}
Given a canonical network $\C$ on $\X$ and some $\Y \subseteq \X$ with $| \Y | \geq 2$, the \emph{canonical subnetwork} of $\C$ induced by $\Y$ is the canonical network $\C||_\Y$ obtained from $\C$ by first taking the subnetwork of $\C$ induced by $\Y$, followed by contracting every 3-blob and 4-blob to a single vertex, and by contracting every cycle contraction edge.
\end{definition}

Indeed, this canonical subnetwork is equivalent to first inducing a subnetwork of the original network, and then creating the canonical form. This follows from the fact that every 3-blob or 4-blob in the canonical subnetwork of a canonical network must have been a triangle or a 4-cycle in the subnetwork of the original network. These blobs should thus be identified by a single vertex. Furthermore, the cycle contraction edges in $\C|_\Y$ correspond to cycle contraction edges in the subnetwork $\N|_\Y$ of the original binary network, since - by definition - the cycle contraction edges in $\C|_\Y$ are only between degree-3 vertices. We present the formal statement in the following observation, yet omit its rather straightforward proof which follows from the reasoning above. Throughout the rest of this section we will frequently, yet implicitly, use this observation.

\begin{observation}
Given a semi-directed level-1 network $\N$ on $\X$ and its canonical form $\C$, if $\Y \subseteq \X$ with $|\Y | \geq 2$, then $(\N|_\Y)^c = \Cr{\Y}$.
\end{observation}


\subsection{Reconstructing a canonical network}
Our strategy to construct a canonical network will be similar to the strategy that was used to construct trees-of-blobs. However, instead of contracting parts of the intermediate trees-of-blobs (which signifies the addition of a reticulation), we add a reticulation. The reduction in the number of used quarnet-splits from $\bigO( n^3)$ to $\bigO ( n \log n)$ has two reasons. On the one hand, we get a reduction of $\Theta (n)$ by using \cref{thm:determine_split} to our advantage. Specifically, we know where the reticulations are in a canonical network (except for those in triangles and 4-cycles) and we can thus find a constant size distinguishing set. Getting the logarithmic factor requires some more care but is loosely based on recursively picking cut-edges that create `balanced' cuts in the network. We now start by transferring the concept of a stem to canonical networks.

Since the tree-of-blobs of a canonical network is equal to the tree-of-blobs of the original network, we can generalize the notions from \cref{def:active_edge} to canonical networks in a straightforward way. To this end, let $\C$ be a canonical network on $\X$ and let $\Y \subset \X$ with $x \in \X \setminus \Y$. Then, a cut-edge of $\Cr{\Y}$ is a \emph{strong stem edge / weak stem edge / pointing edge (for $x$)} if the unique edge that induces the same split in its tree-of-blobs $\T (\Cr{\Y}) = \T (\N |_{\Y})$ is a strong stem edge / weak stem edge / pointing edge for $x$. Similarly, we let a blob in $\Cr{\Y}$ be a \emph{stem blob (for $x$)} if the vertex that represents it in its tree-of-blobs is a stem vertex for $x$. Specifically, if such a blob is a single vertex it is a \emph{stem vertex (for $x$)}, while it is a \emph{stem cycle (for $x$)} if the blob is a cycle. A \emph{cut-path} is a path that only contains cut-edges of $\Cr{\Y}$, and such a cut-path is a \emph{weak stem path (for $x$)} if the corresponding path in its tree-of-blobs is a weak stem subtree for $x$. As was the case for trees-of-blobs, we often omit `for $x$' if the leaf $x$ is clear from the context. We can now define the analogue of a stem for canonical networks.
\begin{definition}[Canonical stem]\label{def:network_stem}
Let $\N$ be a semi-directed level-1 network on $\X$ and let $\Y \subset \X$. Given the canonical network $\Cr{\Y}$ and some $x \in \X \setminus \Y$, a \emph{canonical stem (for $x$)} of $\Cr{\Y}$ is a strong stem edge for $x$, stem vertex for $x$, stem cycle for $x$, or weak stem path for $x$.
\end{definition}

We refer to \cref{fig:network_stem} for examples of the different types of canonical stems. Analogous to \cref{lem:unique_tree_stem}, the next lemma shows that a canonical network has a unique canonical stem and that all pointing edges are oriented towards it. The lemma also details how to attach a leaf $x$ to a stem, for which we need the following new attachment operation. When we \emph{attach $x$ to two different vertices $u$ and $v$}, we mean that we add a new reticulation vertex $w$ with directed edges $uw$ and $vw$, and then attach $x$ to $w$. The different ways to attach the leaves can also be seen in \cref{fig:network_stem}. Note that the lemma does not cover the case when the canonical stem is a stem cycle. Then, we need to attach the leaf somewhere on the cycle, but this cannot be done using $\bigO(1)$ quarnet-splits. Later in this section in \cref{lem:close_cycle} we will solve this issue by deleting the reticulation of the cycle, attaching the leaf, and then putting the reticulation back in. 

\begin{figure}[htb]
\centering
\input{fig_canonic_alg}
\caption{A semi-directed level-1 network $\N$ on $\X = \{ 1, \ldots , 10 \}$. On the right, we have the canonical network $\C$ and five canonical forms of subnetwork $\Cr{\X_i}$, with $\X_i$ the set of the first $i$ leaves of $\X$. In each of the five intermediate canonical networks the oriented edges (with the arrowhead in the middle) indicate a pointing edge. The canonical stems of the canonical networks are highlighted and are in order: a stem vertex of degree-3 in thick blue, a stem vertex of degree-4 in thick blue, a stem cycle in thick blue, a strong stem edge in thick blue, and a weak stem path in thick red. Note that the dashed edges with the arrowhead at the end of the edge are reticulation edges and not pointing edges.}
\label{fig:network_stem}
\end{figure}

\begin{lemma}\label{lem:unique_network_stem}
Let $\N$ be a semi-directed level-1 network on $\X$ and let $\Y \subset \X$. Given the canonical network $\Cr{\Y}$ and some $x \in \X \setminus \Y$, there is a unique canonical stem and the orientation of any pointing edge in $\Cr{\Y}$ is towards this canonical stem. Furthermore,
\begin{enumerate}[label={(\alph*)},noitemsep,topsep=0pt]
    \item if the canonical stem is a strong stem edge, then the canonical network $\Cr{ \Y \cup \{ x\}}$ can be obtained from $\Cr{\Y}$ by attaching $x$ to this edge;
    \item if the canonical stem is a weak stem path, then the canonical network $\Cr{ \Y \cup \{ x\}}$ can be obtained from $\Cr{\Y}$ by attaching $x$ to both end vertices of this path;
    \item if the canonical stem is a stem vertex of degree-3, the canonical network $\Cr{ \Y \cup \{ x\}}$ can be obtained from $\Cr{\Y}$ by attaching $x$ to this vertex; 
    \item if the canonical stem is a stem vertex of degree-4, the canonical network $\Cr{ \Y \cup \{ x\}}$ can be obtained from $\Cr{\Y}$ in constant time using at most five quarnet-splits of $\N$. 
\end{enumerate}
\end{lemma}
\begin{proof}
From \cref{lem:unique_tree_stem} we know that the tree-of-blobs $ \T (\N|_{\Y}) = \T (\Cr{\Y})$ contains a unique stem and all other edges are pointing edges oriented towards this stem. If the stem in this tree-of-blobs is a strong stem edge, $\Cr{\Y}$ has a strong stem edge. Similarly, if this stem is a stem vertex, $\Cr{\Y}$ has a stem vertex or a stem cycle. Lastly, if this stem is a weak stem subtree $T$, we know from \cref{lem:unique_tree_stem}b that $ \T (\N|_{\Y \cup \{x\}}) = \T (\Cr{ \Y \cup \{ x\}})$ can be obtained by contracting $T$ and then attaching $x$ to it. To not contradict the fact that $\Cr{\Y}$ is level-1, all vertices of $T$ must represent trivial blobs in $\N|_\Y$ (and thus in $\Cr{\Y}$) and they form a cut-path. Therefore, there is a weak stem path in $\Cr{\Y}$. This proves the first part of the lemma.

Part (a) (resp. part (c)) of the lemma follows directly from \cref{lem:unique_tree_stem}a (resp. \cref{lem:unique_tree_stem}c) and the fact that triangles (resp. 4-cycles) are suppressed in canonical networks.

For part (b), we already know by the previous reasoning that all vertices of the weak stem path become part of the same cycle $\Cycle$ when we attach $x$. Thus, $x$ is the first leaf (with respect to all other leaves in $\Y$) below the reticulation of $\Cycle$. For canonical networks, we do not care about the ordering of the two leaves in a reticulation cycle pair. Clearly, we can then just attach $x$ to both end vertices of the path (and direct the two new edges towards $x$) to create the cycle $\Cycle$.

For part (d), we let $Y_1|Y_2|Y_3|Y_4$ be the partition of $\Y$ induced by the stem vertex $v$ and we also pick one arbitrary leaf $y_i$ from every $Y_i$. The stem vertex $v$ represents a 4-cycle in $\N|_\Y$ and \cref{lem:unique_tree_stem}c tells us that $v$ becomes a 5-blob in $\Cr{ \Y \cup \{ x\}}$ with induced partition $Y_1|Y_2|Y_3|Y_4|\{x\}$. Such a 5-blob will be a triangle with five cut-edges incident to it in $\Cr{ \Y \cup \{ x\}}$. Then, there will be exactly one subset of four leaves from $\{ y_1, y_2, y_3, y_4, x\}$ that induces a quarnet-split. The remaining fifth leaf then corresponds to the subset of leaves below the reticulation of $\Cycle$, while the quarnet-split is enough to order the other leaves around the triangle.
\end{proof}

The next lemma, which follows from \cref{thm:determine_split}, is crucial for finding out whether a cut-edge is a weak/strong stem edge or a pointing edge. In contrast with the linear amount of quarnet-splits that was needed in the previous section, we need only a constant amount in the case of level-1 networks.
\begin{lemma}\label{cor:find_active_passive}
Let $\N$ be a semi-directed level-1 network on $\X$ and let $\Y \subset \X$ contain $k$ leaves. Given a cut-edge $uv$ in the canonical network $\Cr{\Y}$ and some $x \in \X \setminus \Y$, one can determine in $\bigO(k)$ time whether $uv$ is a strong stem edge, a weak stem edge, or a pointing edge, using at most four quarnet-splits of $\N$.
\end{lemma}
\begin{proof}
If the cut-edge $uv$ induces a split $A|B$ with $|B|=1$, we know that $A \cup \{ x\}|B$ is always a split. If $|B| \geq 2$, we can use \cref{thm:determine_split} to determine whether $A \cup \{ x\}|B$ is a split in $\bigO(k)$ time, using at most two quarnet-splits of $\N$. To be precise, we can always pick the set $B'$ in $\bigO(k)$ time such that it has size at most two. This follows from the fact that, except for the triangles and 4-cycles, we always know where the reticulation vertices are in a canonical network. Finally, reversing the roles of $A$ and $B$ allows us to determine what type $uv$ is.
\end{proof}

Whenever we use the previous lemma and find out that a given edge is a strong stem edge, \cref{lem:unique_network_stem} describes how to add the leaf to the canonical network. If a weak stem edge is found, we first need to find the complete weak stem path before we can use \cref{lem:unique_network_stem} to attach a new leaf. The following lemma proves that this can be done using a logarithmic number of quarnet-splits, making it the first puzzle piece to get the number of quarnet-splits that our algorithm uses down to $\bigO (n \log n)$. \cref{fig:find_passive_cutpath} is a visualization of part of the algorithm described in the proof of the lemma.

\begin{lemma}\label{lem:find_passive_cutpath}
Let $\N$ be a semi-directed level-1 network on $\X$ and let $\Y \subset \X$ contain $k$ leaves. Given a weak stem edge in the canonical network $\Cr{\Y}$ and some $x \in \X \setminus \Y$, one can find the weak stem path in $\Cr{\Y}$ in $\bigO(k)$ time using $\bigO(\log k)$ quarnet-splits of $\N$.
\end{lemma}
\begin{proof}
Let $u_1 u_2$ be the given weak stem edge. By \cref{lem:unique_network_stem} there exists a unique cut-path $P$ that is a weak stem path and that contains $u_1 u_2$. By the same lemma, we also know that we can obtain $\Cr{\Y \cup \{x\}}$ from $\Cr{\Y}$ by attaching $x$ to the end vertices of~$P$. Now let $T$ be the unique maximal subtree of $\Cr{\Y}$ that contains $u_1 u_2$ and does not contain degree-4 vertices, leaves, or vertices that are part of a cycle. Then, $T$ is a binary tree and $P$ is completely in $T$. Otherwise, attaching $x$ to the end vertices of $P$ would not result in a canonical network. \cref{fig:find_passive_cutpath} shows an example of such a tree $T$ for a given weak stem edge. Note that finding $T$ can be done in $\bigO(k)$ time. We also store for each edge of $T$ what its corresponding split is in $\Cr{\Y}$. This allows us to say that an edge in $T$ (or in any of its subtrees) is a strong stem edge / weak stem edge / pointing edge, when we mean that the corresponding edge in $\Cr{\Y}$ is a strong stem edge / weak stem edge / pointing edge.

The algorithm now continues as follows. We delete the edge $u_1 u_2$ from $T$ to create two smaller subtrees $T_1$ and $T_2$ with $u_1$ in $T_1$ and $u_2$ in $T_2$. Clearly, one end vertex $p_1$ of $P$ must be in $T_1$, and the other end vertex $p_2$ in $T_2$. We will now show that one can recursively locate the vertex $p_i$ in the tree $T_i$ in $\bigO( k_i)$ time using $\bigO( \log k_i )$ quarnet-splits, where $k_i$ is the number of vertices in $T_i$. Since knowing $p_1$ and $p_2$ is enough to find $P$ in linear time, the result will then follow.

Whenever $k_i = 1$, we trivially know that the only vertex in $T_i$ is $p_i$ and we are done. Otherwise, if $k_i \geq 2$, we first find a \emph{centroid} $v$ of $T_i$: a vertex such that every component of $T_i - \{ v\}$ contains at most $\frac23 k_i$ vertices. A centroid always exists and can be found in linear time (see e.g. \cite{kannan1996determining,bodlaender1995}). Since $T_i$ is binary and we are in the case with $k_i \geq 2$, there is one edge $vw$ incident to $v$ such that both components of $T_i - \{vw\}$ have at most $\frac23 k_i$ vertices. We will now describe how to find out which of the two components contains $p_i$ and which component we should thus recurse on. This will then prove correctness since both components contain at most $\lfloor \frac23 k_i \rfloor < k_i$ vertices. To this end, we test with \cref{cor:find_active_passive} whether $vw$ is a weak stem edge or a pointing edge in linear time using at most four quarnet-splits. Note that by the uniqueness of the stem (see \cref{lem:unique_network_stem}) $vw$ cannot be a strong stem edge. If $vw$ is a pointing edge, we know from \cref{lem:unique_network_stem} that its corresponding edge in $\Cr{\Y}$ is oriented towards $P$ (and thus $p_i$). Thus, $p_i$ is in the component of $T_i - \{vw\}$ to which $vw$ is pointing. On the other hand, if $vw$ is a weak stem edge, $p_i$ can only be in the component of $T_i - \{vw\}$ that does not contain $u_i$ (or, if we have already recursed and $u_i$ is not in the tree any more, the component furthest away from $u_i$ in the original tree $T$).

We now let $D(k')$ be the recursion depth of this algorithm when applied to a tree with $k'$ vertices. From the above analysis it follows that $D(k') \leq D \left( \lfloor \frac23 k' \rfloor \right) + 1$ with $D(2) = 1$, which results in $D(k_i) = \bigO (\log k_i)$. Since we use $\bigO (k_i)$ time and at most four quarnet-splits per recursive level, our algorithm to find $p_i$ takes $\bigO( k_i)$ time and $\bigO( \log k_i)$ quarnet-splits.
\end{proof}

\begin{figure}[htb]
\centering
\tikzstyle{special_edge}=[-, draw=red, thick]

\begin{tikzpicture}[scale=0.55]
	\begin{pgfonlayer}{nodelayer}
		\node [style={internal_node}] (0) at (-3.75, 5.75) {};
		\node [style={internal_node}] (1) at (-2.75, 5) {};
		\node [style={internal_node}] (2) at (-1.75, 4.25) {};
		\node [style={internal_node}] (3) at (0, 4.25) {};
		\node [style={internal_node}] (4) at (0, 5.25) {};
		\node [style={internal_node}] (5) at (-2.75, 3.25) {};
		\node [style={internal_node}] (6) at (-3.75, 3) {};
		\node [style={internal_node}] (7) at (-2.75, 2.25) {};
		\node [style={internal_node}] (8) at (1.75, 4.25) {};
		\node [style={internal_node}] (9) at (2.75, 4.25) {};
		\node [style={leaf_node}] (10) at (-4.75, 6.5) {};
		\node [style={leaf_node}] (11) at (-4.25, 5) {};
		\node [style={leaf_node}] (12) at (-3.25, 6.5) {};
		\node [style={leaf_node}] (13) at (-2.25, 5.75) {};
		\node [style={leaf_node}] (14) at (-0.5, 6) {};
		\node [style={leaf_node}] (15) at (0.5, 6) {};
		\node [style={leaf_node}] (16) at (-4.5, 2.5) {};
		\node [style={leaf_node}] (17) at (-2.75, 1.25) {};
		\node [style={leaf_node}] (18) at (-2, 1.75) {};
		\node [style={leaf_node}] (19) at (-4.5, 3.5) {};
		\node [style={leaf_node}] (20) at (3.5, 4.75) {};
		\node [style={leaf_node}] (21) at (3.5, 3.75) {};
		\node [style={internal_node}] (22) at (1.75, 3.25) {};
		\node [style={internal_node}] (23) at (2.5, 2.5) {};
		\node [style={internal_node}] (24) at (3.5, 2) {};
		\node [style={internal_node}] (25) at (5, 2) {};
		\node [style={leaf_node}] (26) at (1, 2.5) {};
		\node [style={leaf_node}] (27) at (2, 1.75) {};
		\node [style={leaf_node}] (28) at (3.5, 1.25) {};
		\node [style={internal_node}] (29) at (6, 3) {};
		\node [style={internal_node}] (30) at (7, 2) {};
		\node [style={internal_node}] (31) at (6, 1) {};
		\node [style={leaf_node}] (32) at (5.5, 0.25) {};
		\node [style={leaf_node}] (33) at (6.5, 0.25) {};
		\node [style={leaf_node}] (34) at (8, 2) {};
		\node [style={internal_node}] (35) at (6.75, 3.75) {};
		\node [style={internal_node}] (36) at (6.75, 4.75) {};
		\node [style={leaf_node}] (37) at (6.25, 5.5) {};
		\node [style={leaf_node}] (38) at (7.25, 5.5) {};
		\node [style={leaf_node}] (39) at (7.75, 4) {};
		\node [style={leaf_node}] (40) at (5.25, 3.75) {};
		\node [style={main_label}] (41) at (4, 6.75) {$\C$};
	\end{pgfonlayer}
	\begin{pgfonlayer}{edgelayer}
		\draw [style=edge] (10) to (0);
		\draw [style=edge] (0) to (11);
		\draw [style=edge] (0) to (12);
		\draw [style=edge] (0) to (1);
		\draw [style=edge] (1) to (13);
		\draw [style=edge] (4) to (14);
		\draw [style=edge] (4) to (15);
		\draw [style=edge] (2) to (5);
		\draw [style=edge] (19) to (6);
		\draw [style=edge] (6) to (16);
		\draw [style=edge] (7) to (17);
		\draw [style=edge] (7) to (18);
		\draw [style=edge] (9) to (20);
		\draw [style=edge] (9) to (21);
		\draw [style=edge] (22) to (26);
		\draw [style=edge] (23) to (27);
		\draw [style=edge] (24) to (28);
		\draw [style=edge] (24) to (25);
		\draw [style=edge] (29) to (40);
		\draw [style=edge] (29) to (35);
		\draw [style=edge] (35) to (39);
		\draw [style=edge] (35) to (36);
		\draw [style=edge] (36) to (37);
		\draw [style=edge] (36) to (38);
		\draw [style=edge] (31) to (32);
		\draw [style=edge] (31) to (33);
		\draw [style=edge] (30) to (34);
		\draw [style=edge, bend right=45] (29) to (25);
		\draw [style=edge, bend right=45] (25) to (31);
		\draw [style={ret_arc}, bend left=45] (29) to (30);
		\draw [style={ret_arc}, bend right=45] (31) to (30);
		\draw [style=passive] (3) to (8);
		\draw [style={special_edge}] (8) to (22);
		\draw [style={special_edge}] (8) to (9);
		\draw [style={special_edge}] (22) to (23);
		\draw [style={special_edge}] (23) to (24);
		\draw [style={special_edge}] (1) to (2);
		\draw [style={special_edge}] (2) to (3);
		\draw [style={special_edge}] (3) to (4);
		\draw [style=edge] (6) to (7);
		\draw [style={ret_arc}] (6) to (5);
		\draw [style={ret_arc}] (7) to (5);
	\end{pgfonlayer}
\end{tikzpicture}
\caption{A canonical network $\C$ with unlabeled leaves (indicated by the filled black vertices) and a given weak stem edge in thick red. Together with the other red edges, they induce the binary subtree $T$ that has to contain the weak stem path. This tree is used in the algorithm from \cref{lem:find_passive_cutpath}. }
\label{fig:find_passive_cutpath}
\end{figure}

The following result is a consequence of \cref{lem:unique_network_stem}, which states that pointing edges are always oriented towards the canonical stem. Given two canonical networks $\C_1$ and $\C_2$ on leaf sets $\X_1$ and $\X_2$ with $x_1 \in \X_1$ and $x_2 \in \X_2$ such that $\X_1 \cap \X_2 \subseteq \{x_1, x_2 \}$, we first define the following operation. The graph obtained by \emph{glueing $\C_1$ at $x_1$ to $\C_2$ at $x_2$} is the union of $\C_1$ and $\C_2$, where $x_1$ and $x_2$ are identified as a single vertex and the resulting degree-2 vertex is suppressed. We refer to the left and middle part of \cref{fig:glueing} to get some intuition for the glueing operation and the lemma.

\begin{lemma}\label{cor:glue_restrictions}
Let $\N$ be a semi-directed level-1 network on $\X$ and let $\Y \subset \X$. Let $uv$ be a cut-edge in the canonical network inducing the split $A|B$ and let $x \in \X \setminus \Y$, $a \in A$, and $b\in B$. If $uv$ is a pointing edge for $x$ with orientation $uv$, then $\Cr{\Y \cup \{x\}}$ can be constructed by glueing $\Cr{A \cup \{b\}}$ at $b$ to $\Cr{B \cup \{a, x\}}$ at $a$.
\end{lemma}

\begin{figure}[htb]
\centering


\tikzstyle{special_edge3}=[-]
\tikzstyle{special_edge4}=[-, draw=red, thick]
\tikzstyle{special_edge5}=[draw=gray]
\tikzstyle{special_edge6}=[densely dashed, -{Latex[scale=.9]}, draw=gray]
\tikzstyle{special_edge7}=[densely dashed, draw=gray]
\tikzstyle{special_edge8}=[-, draw=red, ultra thick]
\tikzstyle{special_node2}=[circle, draw=gray, fill=gray, scale=0.275]
\tikzstyle{special_node3}=[circle, draw=gray, fill=white, scale=0.2]
\tikzstyle{semi-active}=[decoration={markings, mark=at position .5 with {\arrow{to}}},postaction={decorate}, ultra thick]

\begin{tikzpicture}[scale=0.5]
	\begin{pgfonlayer}{nodelayer}
		\node [style=none] (207) at (16, -2.25) {};
		\node [style=none] (208) at (16, -3.75) {};
		\node [style=none] (209) at (17.75, -3.75) {};
		\node [style=none] (210) at (17.75, -2.25) {};
		\node [style=none] (233) at (4.75, -1.5) {};
		\node [style=none] (234) at (4.75, -6.25) {};
		\node [style={small_label}] (235) at (5.25, -5.75) {$B$};
		\node [style={small_label}] (236) at (4.25, -5.75) {$A$};
		\node [style={internal_node}] (237) at (3.5, -3.25) {};
		\node [style={leaf_node}, label={left:1}] (238) at (1, -2.5) {};
		\node [style={leaf_node}, label={left:2}] (239) at (1, -4) {};
		\node [style={leaf_node}, label={left:3}] (240) at (3.5, -4.5) {};
		\node [style={leaf_node}, label={left:4}] (241) at (6.75, -5.25) {};
		\node [style={leaf_node}, label={right:5}] (242) at (7.75, -5) {};
		\node [style={medium_label}] (243) at (2.75, -1.25) {$\Cr{\X_7}$};
		\node [style={leaf_node}, label={right:6}] (244) at (7.75, -1.5) {};
		\node [style={internal_node}] (245) at (2, -3.25) {};
		\node [style={leaf_node}, label={left:7}] (246) at (6.75, -1.25) {};
		\node [style={internal_node}] (250) at (7, -2.25) {};
		\node [style={internal_node}] (251) at (7, -4.25) {};
		\node [style={internal_node}] (252) at (5.5, -3.25) {};
		\node [style={internal_node}] (253) at (14.25, -3.25) {};
		\node [style={leaf_node}, label={left:1}] (254) at (11.75, -2.5) {};
		\node [style={leaf_node}, label={left:2}] (255) at (11.75, -4) {};
		\node [style={leaf_node}, label={left:3}] (256) at (14.25, -4.5) {};
		\node [style={internal_node}] (257) at (12.75, -3.25) {};
		\node [style={leaf_node}, label={above:4}] (258) at (16.5, -3.25) {};
		\node [style={medium_label}] (259) at (13.5, -1.25) {$\Cr{A \cup \{4\}}$};
		\node [style={leaf_node}, label={left:4}] (261) at (20, -5.25) {};
		\node [style={leaf_node}, label={right:5}] (262) at (21, -5.25) {};
		\node [style={leaf_node}, label={right:6}] (263) at (22.5, -3.25) {};
		\node [style={leaf_node}, label={right:7}] (264) at (21, -1.25) {};
		\node [style={internal_node}] (265) at (19.5, -3.25) {};
		\node [style={internal_node}] (266) at (20.5, -4.25) {};
		\node [style={internal_node}] (267) at (21.5, -3.25) {};
		\node [style={leaf_node}, label={left:8}] (268) at (20, -1.25) {};
		\node [style={internal_node}] (269) at (20.5, -2.25) {};
		\node [style={leaf_node}, label={above:1}] (270) at (17.25, -3.25) {};
		\node [style={medium_label}] (271) at (20.5, 0) {$\Cr{B \cup \{1, 8 \}}$};
		\node [style={medium_label}] (272) at (17, -4.5) {\textcolor{red!40}{$\Cr{\X_8}$}};
		\node [style={leaf_node}, label={left:4}] (273) at (29, -5.25) {};
		\node [style={leaf_node}, label={right:5}] (274) at (30, -5.25) {};
		\node [style={leaf_node}, label={right:6}] (275) at (31.5, -3.25) {};
		\node [style={leaf_node}, label={right:7}] (276) at (30, -1.25) {};
		\node [style={internal_node}] (278) at (29.5, -4.25) {};
		\node [style={internal_node}] (279) at (30.5, -3.25) {};
		\node [style={leaf_node}, label={left:8}] (280) at (29, -1.25) {};
		\node [style={internal_node}] (281) at (29.5, -2.25) {};
		\node [style={special_node3}] (283) at (28.5, -3.25) {};
		\node [style={special_node2}, label={left:\textcolor{gray}{1}}] (284) at (27.25, -3.25) {};
		\node [style={medium_label}] (285) at (29.5, 0) {$\Cr{B \cup \{ \textcolor{gray}{1}, 8 \}}$};
		\node [style={small_label}] (286) at (6.5, -3.25) {$\Cycle$};
	\end{pgfonlayer}
	\begin{pgfonlayer}{edgelayer}
		\draw [style={special_edge7}] (233.center) to (234.center);
		\draw [style=edge] (238) to (245);
		\draw [style=edge] (245) to (239);
		\draw [style=edge] (245) to (237);
		\draw [style=edge] (240) to (237);
		\draw [style=edge] (250) to (246);
		\draw [style=edge] (250) to (244);
		\draw [style=edge] (250) to (251);
		\draw [style=edge] (251) to (241);
		\draw [style=edge] (251) to (242);
		\draw [style={ret_arc}] (250) to (252);
		\draw [style={ret_arc}] (251) to (252);
		\draw [style=semi-active] (237) to (252);
		\draw [style=edge] (254) to (257);
		\draw [style=edge] (257) to (255);
		\draw [style=edge] (257) to (253);
		\draw [style=edge] (256) to (253);
		\draw [style=edge] (253) to (258);
		\draw [style=edge, bend left] (269) to (267);
		\draw [style=edge, bend left] (267) to (266);
		\draw [style={ret_arc}, bend left] (266) to (265);
		\draw [style={ret_arc}, bend right] (269) to (265);
		\draw [style=edge] (261) to (266);
		\draw [style=edge] (262) to (266);
		\draw [style=edge] (267) to (263);
		\draw [style=edge] (269) to (264);
		\draw [style=edge] (269) to (268);
		\draw [style=edge] (265) to (270);
		\draw [style=edge] (281) to (276);
		\draw [style=edge] (281) to (280);
		\draw [style={special_edge6}, bend right] (281) to (283);
		\draw [style={special_edge6}, bend left] (278) to (283);
		\draw [style={special_edge5}] (284) to (283);
		\draw [style={special_edge3}] (275) to (279);
		\draw [style={special_edge3}] (278) to (274);
		\draw [style={special_edge3}] (278) to (273);
		\draw [style={special_edge4}, bend left] (281) to (279);
		\draw [style={special_edge8}, bend left] (279) to (278);
	\end{pgfonlayer}

 \draw[black!40,fill=red!7,rounded corners=2mm,opacity=0.35] ($(207.center)$) -- ($(208.center)$) -- ($(209.center)$) -- ($(210.center)$) -- cycle;

\end{tikzpicture}
\caption{\emph{Left:} The canonical subnetwork of the canonical network $\C$ from \cref{fig:network_stem} induced by its first seven leaves $\X_7$. The thick cut-edge is a pointing edge for leaf $8$ and induces the split $A|B$ with $A = \{1, 2, 3\}$ and $B = \{4, 5, 6, 7\}$. \emph{Middle:} The canonical subnetworks $\Cr{A \cup \{4\}}$ and $\Cr{B \cup \{1, 8\}}$. Glueing 
$\Cr{A \cup \{4\}}$ at $4$ to $\Cr{B \cup \{1, 8\}}$ at $1$ is done by identifying the leaves in the shaded box and then suppressing the resulting degree-2 vertex. This results in the canonical subnetwork of $\C$ induced by its first eight leaves $\X_8$. \emph{Right:} The canonical subnetwork $\Cr{B \cup \{8\}}$ (excluding the grey edges and vertices) with the weak stem path for leaf $1$ in red. Including the grey edges and vertices produces $\Cr{B \cup \{1, 8\}}$. Note that he thick red cut-edge in $\Cr{B \cup \{8\}}$ corresponds to the single non-reticulation edge of the cycle $\Cycle$ in $\Cr{\X_7}$, while the thinner red cut-edge is one of the four cut-edges incident to the end-points of the thick red cut-edge.}
\label{fig:glueing}
\end{figure}

Observe that to get $\Cr{B\cup \{x, a\}}$ in the lemma, we only need to attach $x$ to $\Cr{B \cup \{a\}}$, which in turn can be obtained from $\Cr{\Y}$ by identifying all vertices on the `$A$-side' of the cut-edge by the single leaf $a$. Thus, the lemma allows us to first cut off part of the network, then attach the new leaf to the remaining part, and afterwards glue the two parts together again. 

In the case where $uv$ points directly towards a reticulation $v$ of a cycle $\Cycle$ (as is the case in \cref{fig:glueing}), we can even get rid of the leaf $a$. That is, instead of attaching $x$ to $\Cr{B \cup \{a\}}$, we can first attach $x$ to $\Cr{B}$ and afterwards attach $a$ to its stem again to obtain $\Cr{B\cup \{x, a\}}$ (see e.g. the right part of \cref{fig:glueing}). Note that the cycle $\Cycle$ does not appear any more in $\Cr{B}$ since none of the leaves below its reticulation are in $B$. Thus, in some sense, we are `unfolding' the cycle $\Cycle$. This also solves the problem of attaching a leaf to a stem cycle mentioned in the discussion before \cref{lem:unique_network_stem}. We prove that this process can be done efficiently in the following lemma and refer to \cref{fig:glueing} for an illustration of parts of the proof.

\begin{lemma}\label{lem:close_cycle}
Let $\N$ be a semi-directed level-1 network on $\X$ and let $\Y \subseteq \X$ contain $k$ leaves. Let $uv$ be a cut-edge in the canonical network $\Cr{\Y}$ inducing the split $A|B$ and let $x \in \X \setminus \Y$, $a\in A$ and $b\in B$. If $uv$ is a pointing edge for $x$ with orientation $uv$ and $v$ is the reticulation of a cycle, then $\Cr{B \cup \{a, x\}}$ can be constructed from $\Cr{B \cup \{x\}}$ in $\bigO(k)$ time using $\bigO (1)$ quarnet-splits of $\N$.
\end{lemma}
\begin{proof}
Let $\Cycle$ be the $c$-cycle (with $c \geq 3$) in $\Cr{\Y}$ of which $v$ is the reticulation and note that $\Cycle$ (and its reticulation~$v$) also exist in $\Cr{B \cup \{a\}}$ (which is just a canonical subnetwork of $\Cr{\Y}$). Since $uv$ is a pointing edge for $x$ in $\Cr{\Y}$, the leaf $x$ is not attached somewhere below the reticulation $v$ when constructing $\Cr{\Y \cup \{x\}}$ from $\Cr{\Y}$ (see \cref{lem:unique_network_stem}). Therefore, to construct $\Cr{B \cup \{a, x\}}$ from $\Cr{B \cup \{x\}}$ we need to attach $a$ such that it creates a cycle with $a$ being the only leaf below its reticulation. The resulting cycle will then be a $c$-cycle or $(c+1)$-cycle (depending on whether $x$ is attached to the cycle or somewhere else). By \cref{lem:unique_network_stem}, this means that we need to attach $a$ to a weak stem path for $a$ in $\Cr{B \cup \{x\}}$. In particular, this weak stem path has to be a path of $c-1$ or $c$ vertices for the new cycle to be a $c$-cycle or $(c+1)$-cycle. We will now describe how to find this weak stem path in $\bigO( k)$ time using $\bigO (1)$ quarnet-splits, which will then prove the lemma.

Note that every non-reticulation edge of $\Cycle$ in $\Cr{B \cup \{a\}}$ (or equivalently in $\Cr{\Y}$) corresponds to a unique cut-edge in $\Cr{B}$ (see e.g. \cref{fig:glueing}). In $\Cr{B \cup \{x\}}$ these cut-edges still exist, although one of them could have been subdivided in the case that $x$ was attached to it. Clearly, for the resulting cycle in $\Cr{B \cup \{a, x\}}$ to be compatible with the cycle $\Cycle$ in $\Cr{B \cup \{a\}}$, all of these cut-edges have to be part of the weak stem path for $a$. Now note that these cut-edges form a path in $\Cr{B \cup \{x\}}$ and that there are either $c-2$ or $c-1$ of them (depending on whether one of them was subdivided by an attachment of $x$). Therefore, we have already found a path of $c-1$ or $c$ vertices containing only weak stem edges. Recall that we know that the weak stem path for $a$ has length at most $c$. Therefore, it is sufficient to check whether the four cut-edges incident to the end-vertices of the path are weak stem edges. By \cref{cor:find_active_passive} this takes $\bigO(k)$ time and $\bigO( 1)$ quarnet-splits for each of the four cut-edges.
\end{proof}

In light of \cref{cor:glue_restrictions} we ideally want to find cut-edges that split the leaf set in a `balanced' way, thus cutting off roughly half of the leaves each time we recurse. This would allow us to get the desired logarithm in the number of quarnet-splits that are used. Unfortunately, such an edge does not always exist (consider e.g. a canonical network with one very large cycle as its only internal blob). However, if we choose our cut-edges smartly, we can use \cref{lem:close_cycle} to `unfold' such a cycle and find a `balanced' cut-edge within the unfolded cycle. To formalize this we will introduce some new `central' structures in canonical networks.

Given a canonical network $\C$ on $\X$, an internal blob $\B$ is a \emph{central blob} if every component of $\C - \B$ contains at most half of the leaves of $\X$. Since a canonical network is obtained from a level-1 network, a central blob is either a \emph{central vertex} or a \emph{central cycle}. Given a central blob, we let a cut-edge $uv$ incident to it be a \emph{central edge}. Clearly, if $uv$ is a central edge inducing the split $A|B$ (where $u$ is in the blob and so $A$ is on the side of the blob) then $|\X| \geq 3$, and so $|A|\geq 2$ and $|B| \leq \frac12 |\X|$. A central edge $uv$ incident to a central blob $\B$ (where $u$ is in $\B$) of a canonical network is \emph{nice} if $uv$ satisfies one of the following two statements: (i) $\B$ is a central vertex, $|A|\leq \frac34 |\X|$ and $|B| \geq 2$; (ii) $\B$ is a central cycle and $uv$ is the cut-edge incident to the reticulation of this cycle. See the left part of \cref{fig:central_edges} for an illustration of central blobs and central edges. Note that a canonical network does not always have both a central vertex and a central cycle.

\begin{figure}[htb]
\centering
\tikzstyle{special_edge}=[-, draw=customred, very thick]
\tikzstyle{special_edge2}=[-, draw=customgreen, thick]
\tikzstyle{special_edge3}=[densely dashed, -{Latex[scale=.7]}, draw=customgreen, thick]
\tikzstyle{special_node}=[circle, draw=customgreen, fill=customgreen, scale=0.45]

\tikzstyle{special_edge4}=[draw=orange, thick]
\tikzstyle{special_edge5}=[draw=gray]
\tikzstyle{special_edge6}=[densely dashed, -{Latex[scale=.9]}, draw=gray]
\tikzstyle{special_edge7}=[draw=orange, densely dotted, very thick]
\tikzstyle{special_node2}=[circle, draw=gray, fill=gray, scale=0.275]
\tikzstyle{special_node3}=[circle, draw=gray, fill=white, scale=0.2]

\begin{tikzpicture}[scale=0.5]
	\begin{pgfonlayer}{nodelayer}
		\node [style={internal_node}] (0) at (5.5, 3.5) {};
		\node [style={internal_node}] (1) at (9, 3.5) {};
		\node [style={internal_node}] (2) at (6.5, 2) {};
		\node [style={internal_node}] (3) at (8, 2) {};
		\node [style={internal_node}] (4) at (8, 5) {};
		\node [style={internal_node}] (5) at (6.5, 5) {};
		\node [style={leaf_node}, label={right:16}] (7) at (5.75, 6) {};
		\node [style={leaf_node}, label={right:9}] (8) at (5.75, 1) {};
		\node [style={leaf_node}, label={left:10}] (9) at (8, 0.75) {};
		\node [style={leaf_node}, label={right:11}] (10) at (9, 1.25) {};
		\node [style={leaf_node}, label={left:15}] (11) at (8, 6.25) {};
		\node [style={leaf_node}, label={right:14}] (12) at (9, 5.75) {};
		\node [style={internal_node}] (14) at (3.75, 5) {};
		\node [style={leaf_node}, label={above:1}] (15) at (4.5, 5) {};
		\node [style={leaf_node}, label={above:3}] (16) at (3, 5) {};
		\node [style={leaf_node}, label={above:2}] (17) at (3.75, 5.75) {};
		\node [style={internal_node}] (18) at (0.25, 3.5) {};
		\node [style={internal_node}] (19) at (-0.75, 4.25) {};
		\node [style={internal_node}] (20) at (-0.75, 2.75) {};
		\node [style={leaf_node}, label={left:6}] (21) at (-1.5, 2) {};
		\node [style={leaf_node}, label={right:7}] (22) at (-0.5, 1.75) {};
		\node [style={leaf_node}, label={left:5}] (23) at (-1.5, 5) {};
		\node [style={leaf_node}, label={right:4}] (24) at (-0.5, 5.25) {};
		\node [style={small_label}] (25) at (7.25, 3.5) {$\Cycle$};
		\node [style={special_node}, label={below:$v$}] (28) at (3.75, 3.5) {};
		\node [style={internal_node}] (29) at (2, 3.5) {};
		\node [style={leaf_node}, label={right:8}] (30) at (2, 2.25) {};
		\node [style={internal_node}] (31) at (10.25, 3.5) {};
		\node [style={leaf_node}, label={right:13}] (32) at (11, 4.25) {};
		\node [style={leaf_node}, label={right:12}] (33) at (11, 2.75) {};
		\node [style={main_label}] (58) at (1.25, 6.25) {$\C$};
		\node [style={internal_node}] (71) at (21.75, 3.5) {};
		\node [style={internal_node}] (73) at (22.75, 2) {};
		\node [style={internal_node}] (74) at (24.25, 2) {};
		\node [style={internal_node}] (75) at (24.25, 5) {};
		\node [style={internal_node}] (76) at (22.75, 5) {};
		\node [style={leaf_node}, label={right:16}] (77) at (22, 6) {};
		\node [style={leaf_node}, label={right:9}] (78) at (22, 1) {};
		\node [style={leaf_node}, label={left:10}] (79) at (24.25, 0.75) {};
		\node [style={leaf_node}, label={right:11}] (80) at (25.25, 1.25) {};
		\node [style={leaf_node}, label={left:15}] (81) at (24.25, 6.25) {};
		\node [style={leaf_node}, label={right:14}] (82) at (25.25, 5.75) {};
		\node [style={internal_node}] (83) at (20, 5) {};
		\node [style={leaf_node}, label={above:1}] (84) at (20.75, 5) {};
		\node [style={leaf_node}, label={above:3}] (85) at (19.25, 5) {};
		\node [style={leaf_node}, label={above:2}] (86) at (20, 5.75) {};
		\node [style={internal_node}] (87) at (16.5, 3.5) {};
		\node [style={internal_node}] (88) at (15.5, 4.25) {};
		\node [style={internal_node}] (89) at (15.5, 2.75) {};
		\node [style={leaf_node}, label={left:6}] (90) at (14.75, 2) {};
		\node [style={leaf_node}, label={right:7}] (91) at (15.75, 1.75) {};
		\node [style={leaf_node}, label={left:5}] (92) at (14.75, 5) {};
		\node [style={leaf_node}, label={right:4}] (93) at (15.75, 5.25) {};
		\node [style={small_label}] (94) at (23.5, 3.5) {$\Cycle$};
		\node [style={internal_node}] (95) at (20, 3.5) {};
		\node [style={internal_node}] (96) at (18.25, 3.5) {};
		\node [style={leaf_node}, label={right:8}] (97) at (18.25, 2.25) {};
		\node [style={main_label}] (101) at (17.5, 6.25) {$\C$};
		\node [style={special_node3}] (102) at (25.25, 3.5) {};
		\node [style={special_node3}] (103) at (26.5, 3.5) {};
		\node [style={special_node2}, label={right:13}] (104) at (27.25, 4.25) {};
		\node [style={special_node2}, label={right:12}] (105) at (27.25, 2.75) {};
	\end{pgfonlayer}
	\begin{pgfonlayer}{edgelayer}
		\draw [style=edge] (14) to (15);
		\draw [style=edge] (14) to (17);
		\draw [style=edge] (14) to (16);
		\draw [style=edge] (19) to (20);
		\draw [style=edge] (19) to (24);
		\draw [style=edge] (19) to (23);
		\draw [style=edge] (20) to (21);
		\draw [style=edge] (20) to (22);
		\draw [style={ret_arc}] (19) to (18);
		\draw [style={ret_arc}] (20) to (18);
		\draw [style={special_edge2}, bend left=345] (4) to (5);
		\draw [style={special_edge2}, bend left=15] (3) to (2);
		\draw [style={special_edge2}, bend left] (2) to (0);
		\draw [style=edge] (14) to (28);
		\draw [style=edge] (28) to (29);
		\draw [style=edge] (29) to (18);
		\draw [style=edge] (30) to (29);
		\draw [style=edge] (32) to (31);
		\draw [style=edge] (31) to (33);
		\draw [style=edge] (5) to (7);
		\draw [style=edge] (2) to (8);
		\draw [style={special_edge2}, bend right] (5) to (0);
		\draw [style=edge] (3) to (9);
		\draw [style=edge] (3) to (10);
		\draw [style=edge] (11) to (4);
		\draw [style=edge] (4) to (12);
		\draw [style={special_edge3}, bend left] (4) to (1);
		\draw [style={special_edge3}, bend right] (3) to (1);
		\draw [style={special_edge}] (28) to (0);
		\draw [style={special_edge}] (1) to (31);
		\draw [style=edge] (83) to (84);
		\draw [style=edge] (83) to (86);
		\draw [style=edge] (83) to (85);
		\draw [style=edge] (88) to (89);
		\draw [style=edge] (88) to (93);
		\draw [style=edge] (88) to (92);
		\draw [style=edge] (89) to (90);
		\draw [style=edge] (89) to (91);
		\draw [style={ret_arc}] (88) to (87);
		\draw [style={ret_arc}] (89) to (87);
		\draw [style=edge] (83) to (95);
		\draw [style=edge] (95) to (96);
		\draw [style=edge] (96) to (87);
		\draw [style=edge] (97) to (96);
		\draw [style=edge] (76) to (77);
		\draw [style=edge] (73) to (78);
		\draw [style=edge] (95) to (71);
		\draw [style={special_edge4}, bend right] (71) to (73);
		\draw [style={special_edge4}, bend right=15] (73) to (74);
		\draw [style={special_edge4}, bend left=345] (75) to (76);
		\draw [style={special_edge7}, bend right] (76) to (71);
		\draw [style={special_edge6}, bend left] (75) to (102);
		\draw [style={special_edge6}, bend right] (74) to (102);
		\draw [style={special_edge5}] (102) to (103);
		\draw [style={special_edge5}] (103) to (104);
		\draw [style={special_edge5}] (103) to (105);
		\draw [style={special_edge4}] (81) to (75);
		\draw [style={special_edge4}] (82) to (75);
		\draw [style={special_edge4}] (74) to (79);
		\draw [style={special_edge4}] (74) to (80);
	\end{pgfonlayer}
\end{tikzpicture}
\caption{A canonical network $\C$ on $\X = \{1, \ldots, 16 \}$. \emph{Left:} The two central blobs of~$\C$: a central vertex $v$ in pink and a central cycle $\Cycle$ in pink. The left thick blue edge is a nice central edge incident to~$v$, and the right thick blue edge is a nice central edge incident to $\Cycle$. \emph{Right:} The spine edges of the cycle $\Cycle$ in orange with the dotted orange edge being a nice spine edge of $\Cycle$. The filled black leaves are the leaves $\Y = \{1, \ldots, 11, 14, 15, 16\}$ not below the reticulation of $\Cycle$. Deleting the grey edges and vertices forms the canonical subnetwork $\Cr{\Y}$, in which the nice spine edge is a cut-edge.}
\label{fig:central_edges}
\end{figure}

By definition, a cycle $\Cycle$ in a canonical network $\C$ will always contain two degree-4 vertices. We call the four cut-edges incident to the two degree-4 vertices the \emph{outer spine edges} of the cycle and the non-reticulation edges within the cycle the \emph{inner spine edges}. Together, they form the \emph{spine edges} of $\Cycle$ (see \cref{fig:central_edges}). If we let $\Y \subset \X$ be the set of leaves not below the reticulation of $\Cycle$, then every spine edge of $\Cycle$ \emph{corresponds} to a unique cut-edge in $\Cr{\Y}$. The \emph{spine split} $A'|B'$ induced by a spine edge of $\Cycle$ is the split induced by the corresponding cut-edge in $\Cr{\Y}$. Thus, $A'|B'$ is a partition of $\Y$ and not of $\X$. A spine edge of a central cycle is \emph{nice} if it has spine split $A'|B'$ such that $2 \leq |A'| \leq \frac34 |\X|$ and $2 \leq |B'| \leq \frac34 |\X|$ (see \cref{fig:central_edges}). That means that a nice spine edge is `balanced' with respect to the leaves in $\X$ (and not necessarily to those in $\Y \subset \X$). \cref{lem:nice_central_edge} shows that we can efficiently find the different types of nice edges.

\begin{lemma}\label{lem:nice_central_edge}
Given a canonical level-1 network $\C$ on $\X$ with $n\geq 5$ leaves, one can find a central blob $\B$ and a nice central edge incident to $\B$ in $\bigO(n)$ time. Furthermore, if $\B$ is a central cycle, one can find a nice spine edge of $\B$ in $\bigO(n)$ time.
\end{lemma}
\begin{proof}
The following algorithm shows the existence of a central blob and a nice central edge and how to find them. First, create the tree-of-blobs $\T (\C)$ in $\bigO (n)$ time. An internal vertex $v$ of $\T (\C)$ is \emph{leaf-centroid} if every component of $\T (\C) - v$ contains at most $\frac{n}{2}$ leaves. A straightforward greedy-type algorithm from \cite[Lem.\,5]{bodlaender1995} shows that a leaf-centroid always exists and how to find one in linear time. Clearly, an internal vertex of $\T (\C)$ is a leaf-centroid if and only if the blob it represents is a central blob. Thus, we have found a central blob $\B$. If $\B$ is a central cycle with reticulation $r$, we can pick the cut-edge incident to $r$ as our nice central edge. Otherwise, if $\B$ is a central vertex $v$, we pick a cut-edge $vw$ (inducing the split $A|B$) that cuts off most of the leaves from $v$. That is, $|B|$ is as large as possible. Since $v$ has at most four incident cut-edges, we have $|A|\leq \frac34 n$. Furthermore, since $n \geq 5$ and $vw$ cuts off most of the leaves, we must also have $|B| \geq 2$. Thus, $vw$ is a nice central edge.

Now assume that $\B$ is a central cycle and let $Y_1 | Y_2 | \ldots | Y_{k-1} | Y_k | Z$ be the partition of $\X$ induced by $\B$. Assume that $Z$ is the set of leaves below the reticulation, and that the $Y_i$ follow the order of the cycle. Note that $Y_1, Y_2, Y_{k-1}$ and $Y_k$ correspond to the outer spine edges of the cycle and that the order of $Y_1$ and $Y_2$ (resp. $Y_{k-1}$ and $Y_k$) is arbitrary. Then, $\Y = \bigcup_{i=1}^k Y_i$ is the set of leaves not below the reticulation. Note that for every $\ell \in \{ 1, \ldots, k-1\}$ there is a cut-edge in $\Cr{\Y}$ inducing the split $Y_1 \cup \cdots \cup Y_{\ell} | Y_{\ell+1} \cup \cdots \cup Y_k$. In particular, these are all edges corresponding to spine edges of~$\B$ (see \cref{fig:central_edges}).

Finding a nice spine edge of $\B$ is now fairly simple. Since $\B$ was a central cycle, we know that $1 \leq |Y_i| \leq \frac{n}{2}$ for each $i$. Our first candidate is the cut-edge that cuts off $Y_1$. If $|Y_1| \geq \frac{n}{4}$, that will be our nice spine edge. Otherwise, we move to the cut-edge that cuts off $Y_1$ and $Y_2$. Again, if $|Y_1 \cup Y_2| \geq \frac{n}{4}$ we stop, and if $|Y_1 \cup Y_2| < \frac{n}{4}$ we move to the next cut-edge. We keep continuing this procedure until $|Y_1 \cup \cdots \cup Y_{\ell}| \geq \frac{n}{4}$ for some $\ell < k$. Then, the cut-edge $e$ inducing the split $A'|B'$ (where $A' = Y_1 \cup \cdots \cup Y_{\ell}$ and $B'= Y_{\ell+1} \cup \cdots \cup Y_k$) has $|A'|, |B'| \leq \frac34 n$. Furthermore, since $n\geq 5$, this also forces $|A'|, |B'| \geq 2$. Thus, $e$ is a nice spine edge of $\B$.
\end{proof}

Finally, we are ready to devise \cref{alg:canonical_network_recursive} which recursively adds a leaf to a canonical network. In the spirit of \cref{alg:blobtree}, our final reconstruction algorithm will then repeatedly apply \cref{alg:canonical_network_recursive} to attach $n$ leaves (in any arbitrary order) to a canonical network in $\bigO( n^2)$ time using $\bigO (n \log n)$ quarnet-splits. Before proving this in \cref{thm:alg_canonical_network}, we first roughly sketch how \cref{alg:canonical_network_recursive} works.  We refer to \cref{fig:alg_canonical} for an illustration of the recursive part of the algorithm. Apart from the quarnet-splits, the canonical network $\Cr{\Y}$ and the new leaf $x$, the algorithm takes as optional input a cut-edge $uv$. At the top level of the recursion this cut-edge is not provided, so we will leave this variation to the end of our discussion. 

Unless we are in the base case, we first find a central blob $\B$ and incident central edge $uv$, after which we check whether it is a strong stem edge, weak stem edge or pointing edge. If $uv$ turns out to be a strong stem edge, we can attach $x$ to it (see \cref{lem:unique_network_stem,fig:network_stem}). In the case that $uv$ is a weak stem edge, we first find the weak stem path, after which we attach $x$ to it (again see \cref{lem:unique_network_stem,fig:network_stem}). 

The final case is when $uv$ (with induced split $A|B$) is a pointing edge. Then, it is enough to construct $\Cr{B \cup \{x, a\}}$, which can subsequently be glued to $\Cr{A \cup \{b\}}$ to create $\Cr{\Y \cup \{x\}}$ on line \ref{line:glue} (see also the middle part of \cref{fig:glueing}). On the one hand, if $uv$ is oriented away from the blob $\B$ (i.e. $u$ is in $\B$) or $\B$ is a central vertex, we simply attach $x$ to $\Cr{B \cup \{a\}}$ by recursion. Since $uv$ is a nice central edge this will always cut off some leaves. In the other case, $\B$ is a central cycle and $uv$ is oriented towards $\B$ with $v$ the reticulation of $\B$. We cannot recurse directly on $\Cr{B \cup \{a\}}$ since it might not cut off anything (consider again a network with one large cycle). So, we recurse on $\Cr{B}$ and attach $x$. Here, we also provide a cut-edge $st$ which will take over the role of $uv$ in the next level of recursion. Note that we do not provide an incident blob, so in the next level of recursion $\B$ is undefined. After $\Cr{B \cup \{x\}}$ is created, we attach the leaf $a$ again (see \cref{lem:close_cycle} and the right part of \cref{fig:glueing}).

\begin{algorithm}[htb]
\caption{Recursive procedure to attach a leaf to a canonical network}
\label{alg:canonical_network_recursive}
\Input{quarnet-splits of a semi-directed level-1 network $\N$ on $\X$; canonical network $\Cr{\Y}$ on $\Y \subset \X$; leaf $x \in \X \setminus \Y$ to be attached; \emph{optional:} a non-trivial cut-edge $uv$ corresponding to a nice spine edge of the last recursive call}
\Output{canonical network $\Cr{\Y \cup \{x\}}$}

\If{$|\Y| \in \{2, 3, 4\}$}{
    find the canonical stem of $\Cr{\Y}$ by determining for each edge of $\Cr{\Y}$ whether it is a strong stem edge, weak stem edge or pointing edge, using the quarnet-splits of $\N$ and \cref{cor:find_active_passive}\\
    $\Cr{\Y \cup \{x\}}$ is constructed from $\Cr{\Y}$ as described in \cref{lem:unique_network_stem}\\
    
    \Return{$\Cr{\Y \cup \{x\}}$}
}
\If{no cut-edge $uv$ is provided}{
    $\B \gets$ central blob of $\Cr{\Y}$, using \cref{lem:nice_central_edge} \\
    $uv \gets$ nice central edge of $\Cr{\Y}$ incident to $\B$, using \cref{lem:nice_central_edge} \label{line:central}
}
determine whether $uv$ is a strong stem edge, weak stem edge or pointing edge, using the quarnet-splits of $\N$ and \cref{cor:find_active_passive}\\

\If{$uv$ is a strong stem edge}{
    $\Cr{\Y \cup \{x\}}$ is constructed from $\Cr{\Y}$ by attaching $x$ to $uv$ \tcp*{see \cref{fig:network_stem,fig:alg_canonical}}
}

\ElseIf{$uv$ is a weak stem edge}{
    determine the weak stem path in $\Cr{\Y}$ containing $uv$, using the quarnet-splits of $\N$ and \cref{lem:find_passive_cutpath} \\
    $\Cr{\Y \cup \{x\}}$ is constructed from $\Cr{\Y}$ by attaching $x$ to the end vertices of the cut-path \tcp*{see \cref{fig:network_stem}}
}

\ElseIf(\tcp*[f]{otherwise, reverse roles of $u,v$}){$uv$ is a pointing edge with orientation $uv$\label{line:case3}}{
    $A|B \gets$ split induced by $uv$ \tcp*{this implies that $A$ is on the side of $u$}
    $a \gets$ arbitrary leaf from $A$; $b \gets$ arbitrary leaf from $B$ \\

    \If{$\B$ is undefined \KwOr $u$ is in $\B$ \KwOr $\B$ is a central vertex \label{line:subcase1}}{
        $\Cr{B \cup \{x, a\}}$ is constructed from $\Cr{B \cup \{a\}}$ by attaching $x$ with (recursion on) \cref{alg:canonical_network_recursive}
    }
    \Else{
        $st \gets$ non-trivial cut-edge of $\Cr{B}$ corresponding to a nice spine edge of $\B$ in $\Cr{\Y}$, using \cref{lem:nice_central_edge}\\
        $\Cr{B \cup \{x\}}$ is constructed from $\Cr{B}$ by attaching $x$ with (recursion on) \cref{alg:canonical_network_recursive} and providing $st$ as optional argument \label{line:spine_recurse}\\
        $\Cr{B \cup \{x, a\}}$ is constructed from $\Cr{B \cup \{x\}}$ by attaching $a$, using the quarnet-splits of $\N$ and \cref{lem:close_cycle} \label{line:close_cycle}\\
    }

    $\Cr{\Y \cup \{x\}} $ is constructed by glueing $\Cr{A \cup \{b\}}$ at $b$ to $\Cr{B \cup \{a, x\}}$ at $a$ \label{line:glue} \tcp*{see \cref{fig:glueing}}
}

\Return{$\Cr{\Y \cup \{x\}}$}

\end{algorithm}

\begin{figure}[htb]
\centering
\tikzstyle{special_edge2}=[draw=customgreen, thick]
\tikzstyle{special_edge3}=[densely dashed, -{Latex[scale=.7]}, draw=customgreen, thick]
\tikzstyle{special_edge6}=[draw=customgreen, densely dotted, very thick]
\tikzstyle{special_node}=[circle, draw=customgreen, fill=customgreen, scale=0.45]

\tikzstyle{special_node3}=[circle, draw=gray, fill=gray, scale=0.2]
\tikzstyle{special_edge7}=[draw=gray]
\tikzstyle{special_edge8}=[draw=gray]
\tikzstyle{special_node2}=[circle, draw=gray, fill=gray, scale=0.2]

\tikzstyle{special_edge4}=[-{Latex[scale=.9]}, draw=gray!55, dashed]

\tikzstyle{semi-active}=[decoration={markings, mark=at position .7 with {\arrow{to}}},postaction={decorate}, ultra thick, draw=customred]
\tikzstyle{active}=[-, draw=blue, ultra thick]

\begin{tikzpicture}[scale=0.44]
	\begin{pgfonlayer}{nodelayer}
		\node [style={internal_node}] (0) at (7.25, 0) {};
		\node [style={internal_node}] (1) at (10.75, 0) {};
		\node [style={internal_node}] (2) at (8.25, -1.5) {};
		\node [style={internal_node}] (3) at (9.75, -1.5) {};
		\node [style={internal_node}] (4) at (9.75, 1.5) {};
		\node [style={leaf_node}, label={right:9}] (7) at (7.5, -2.5) {};
		\node [style={leaf_node}, label={left:10}] (8) at (9.75, -2.75) {};
		\node [style={leaf_node}, label={right:11}] (9) at (10.75, -2.25) {};
		\node [style={leaf_node}, label={left:15}] (10) at (9.75, 2.75) {};
		\node [style={leaf_node}, label={right:14}] (11) at (10.75, 2.25) {};
		\node [style={internal_node}] (12) at (5.5, 1.5) {};
		\node [style={leaf_node}, label={above:1}] (13) at (6.25, 1.5) {};
		\node [style={leaf_node}, label={above:3}] (14) at (4.75, 1.5) {};
		\node [style={leaf_node}, label={above:2}] (15) at (5.5, 2.25) {};
		\node [style={internal_node}] (16) at (2, 0) {};
		\node [style={internal_node}] (17) at (1, 0.75) {};
		\node [style={internal_node}] (18) at (1, -0.75) {};
		\node [style={leaf_node}, label={left:6}] (19) at (0.25, -1.5) {};
		\node [style={leaf_node}, label={right:7}] (20) at (1.25, -1.75) {};
		\node [style={leaf_node}, label={left:5}] (21) at (0.25, 1.5) {};
		\node [style={leaf_node}, label={right:4}] (22) at (1.25, 1.75) {};
		\node [style={special_node}, label={below:$\B_1$}] (24) at (5.5, 0) {};
		\node [style={internal_node}] (25) at (3.75, 0) {};
		\node [style={leaf_node}, label={right:8}] (26) at (3.75, -1.25) {};
		\node [style={internal_node}] (27) at (12, 0) {};
		\node [style={leaf_node}, label={right:13}] (28) at (12.75, 0.75) {};
		\node [style={leaf_node}, label={right:12}] (29) at (12.75, -0.75) {};
		\node [style={medium_label}] (30) at (5.5, 5) {$\Cr{\Y_1 }$};
		\node [style={internal_node}] (63) at (16.75, 0) {};
		\node [style={internal_node}] (64) at (20.25, 0) {};
		\node [style={internal_node}] (65) at (17.75, -1.5) {};
		\node [style={internal_node}] (66) at (19.25, -1.5) {};
		\node [style={internal_node}] (67) at (19.25, 1.5) {};
		\node [style={leaf_node}, label={right:9}] (70) at (17, -2.5) {};
		\node [style={leaf_node}, label={left:10}] (71) at (19.25, -2.75) {};
		\node [style={leaf_node}, label={right:11}] (72) at (20.25, -2.25) {};
		\node [style={leaf_node}, label={left:15}] (73) at (19.25, 2.75) {};
		\node [style={leaf_node}, label={right:14}] (74) at (20.25, 2.25) {};
		\node [style={small_label}] (75) at (18.5, 0) {$\B_2$};
		\node [style={internal_node}] (77) at (21.5, 0) {};
		\node [style={leaf_node}, label={right:13}] (78) at (22.25, 0.75) {};
		\node [style={leaf_node}, label={right:12}] (79) at (22.25, -0.75) {};
		\node [style={leaf_node}, label={above:1}] (80) at (15.5, 0) {};
		\node [style={medium_label}] (82) at (17, 5) {$\Cr{\Y_2 }$};
		\node [style={internal_node}] (83) at (26.25, 0) {};
		\node [style={internal_node}] (85) at (26.5, -1) {};
		\node [style={internal_node}] (86) at (27.5, -2) {};
		\node [style={internal_node}] (87) at (27.5, 1.75) {};
		\node [style={leaf_node}, label={above:9}] (88) at (25.5, -1.5) {};
		\node [style={leaf_node}, label={left:10}] (89) at (27.25, -3) {};
		\node [style={leaf_node}, label={right:11}] (90) at (28.25, -2.75) {};
		\node [style={leaf_node}, label={left:15}] (91) at (27.25, 2.75) {};
		\node [style={leaf_node}, label={right:14}] (92) at (28.25, 2.5) {};
		\node [style={leaf_node}, label={above:1}] (97) at (25, 0) {};
		\node [style={medium_label}] (99) at (25.5, 5) {$\Cr{\Y_3 }$};
		\node [style={internal_node}] (114) at (32, -0.5) {};
		\node [style={internal_node}] (115) at (32.5, -1.5) {};
		\node [style={leaf_node}, label={left:9}] (116) at (32.25, -2.5) {};
		\node [style={leaf_node}, label={right:10}] (117) at (33.25, -2.25) {};
		\node [style={leaf_node}, label={left:15}] (118) at (32.25, 2.5) {};
		\node [style={leaf_node}, label={right:14}] (119) at (33.25, 2.25) {};
		\node [style={leaf_node}, label={above:1}] (120) at (31, -0.75) {};
		\node [style={medium_label}] (121) at (31.25, 5) {$\Cr{\Y_4}$};
		\node [style=none] (122) at (32, 0.5) {};
		\node [style={special_node2}, label={above:\textcolor{gray}{16}}] (123) at (31, 0.75) {};
		\node [style={special_node}, label={right:$\B_4$}] (124) at (32.5, 1.5) {};
		\node [style=none] (125) at (8.25, 1.5) {};
		\node [style=none] (126) at (17.75, 1.5) {};
		\node [style=none] (127) at (26.5, 1) {};
		\node [style={special_node3}, label={above:\textcolor{gray}{16}}] (128) at (7.5, 2.5) {};
		\node [style={special_node3}, label={above:\textcolor{gray}{16}}] (130) at (17, 2.5) {};
		\node [style={special_node3}, label={above:\textcolor{gray}{16}}] (131) at (25.5, 1.75) {};
		\node [style={small_label}] (132) at (27.75, -1) {$\cancel{\B_3}$};
		\node [style={medium_label}] (133) at (5.5, -4.75) {$\Cr{\Y_1 \textcolor{gray}{\cup \{16\}} }$};
		\node [style={medium_label}] (134) at (17, -4.75) {$\Cr{\Y_2 \textcolor{gray}{\cup \{16\}} }$};
		\node [style={medium_label}] (135) at (25.5, -4.75) {$\Cr{\Y_3\textcolor{gray}{\cup \{16\}} }$};
		\node [style={medium_label}] (136) at (31.25, -4.75) {$\Cr{\Y_4 \textcolor{gray}{\cup \{16\}} }$};
		\node [style=none] (137) at (7.5, 5) {};
		\node [style=none] (138) at (15, 5) {};
		\node [style=none] (139) at (19, 5) {};
		\node [style=none] (140) at (23.5, 5) {};
		\node [style=none] (141) at (19.5, -4.75) {};
		\node [style=none] (142) at (23, -4.75) {};
		\node [style=none] (143) at (8, -4.75) {};
		\node [style=none] (144) at (14.5, -4.75) {};
		\node [style=none] (145) at (27.5, 5) {};
		\node [style=none] (146) at (29.25, 5) {};
		\node [style=none] (147) at (28, -4.75) {};
		\node [style=none] (148) at (28.75, -4.75) {};
		\node [style=none] (149) at (33.5, 5) {};
		\node [style=none] (150) at (33.75, -4.75) {};
	\end{pgfonlayer}
	\begin{pgfonlayer}{edgelayer}
		\draw [style=edge] (12) to (13);
		\draw [style=edge] (12) to (15);
		\draw [style=edge] (12) to (14);
		\draw [style=edge] (17) to (18);
		\draw [style=edge] (17) to (22);
		\draw [style=edge] (17) to (21);
		\draw [style=edge] (18) to (19);
		\draw [style=edge] (18) to (20);
		\draw [style={ret_arc}] (17) to (16);
		\draw [style={ret_arc}] (18) to (16);
		\draw [style=edge] (12) to (24);
		\draw [style=edge] (24) to (25);
		\draw [style=edge] (25) to (16);
		\draw [style=edge] (26) to (25);
		\draw [style=edge] (28) to (27);
		\draw [style=edge] (27) to (29);
		\draw [style=edge] (2) to (7);
		\draw [style=edge] (3) to (8);
		\draw [style=edge] (3) to (9);
		\draw [style=edge] (10) to (4);
		\draw [style=edge] (4) to (11);
		\draw [style={ret_arc}, bend left] (4) to (1);
		\draw [style={ret_arc}, bend right] (3) to (1);
		\draw [style=edge] (1) to (27);
		\draw [style=edge, bend right=15] (2) to (3);
		\draw [style=edge, bend right] (0) to (2);
		\draw [style=semi-active] (24) to (0);
		\draw [style={special_edge2}, bend left] (65) to (63);
		\draw [style=edge] (78) to (77);
		\draw [style=edge] (77) to (79);
		\draw [style=edge] (65) to (70);
		\draw [style=edge] (66) to (71);
		\draw [style=edge] (66) to (72);
		\draw [style=edge] (73) to (67);
		\draw [style=edge] (67) to (74);
		\draw [style={special_edge3}, bend left] (67) to (64);
		\draw [style={special_edge3}, bend right] (66) to (64);
		\draw [style=edge] (63) to (80);
		\draw [style=semi-active] (77) to (64);
		\draw [style={special_edge6}, bend right=15] (65) to (66);
		\draw [style=edge] (85) to (88);
		\draw [style=edge] (86) to (89);
		\draw [style=edge] (86) to (90);
		\draw [style=edge] (91) to (87);
		\draw [style=edge] (87) to (92);
		\draw [style=edge] (83) to (97);
		\draw [style=semi-active, bend left=15] (86) to (85);
		\draw [style=edge, bend right=15] (83) to (85);
		\draw [style=edge] (115) to (116);
		\draw [style=edge] (114) to (120);
		\draw [style=edge, bend right=15] (114) to (115);
		\draw [style=edge] (117) to (115);
		\draw [style=active] (122.center) to (114);
		\draw [style={special_edge8}] (122.center) to (123);
		\draw [style=edge] (118) to (124);
		\draw [style=edge] (124) to (119);
		\draw [style=active, bend right=15] (124) to (122.center);
		\draw [style=edge, bend left=345] (4) to (125.center);
		\draw [style=edge, bend right] (125.center) to (0);
		\draw [style={special_edge2}, bend left=15] (126.center) to (67);
		\draw [style={special_edge2}, bend right] (126.center) to (63);
		\draw [style=edge, bend left=15] (127.center) to (87);
		\draw [style=edge, bend right=15] (127.center) to (83);
		\draw [style={special_edge7}] (128) to (125.center);
		\draw [style={special_edge7}] (130) to (126.center);
		\draw [style={special_edge7}] (131) to (127.center);
		\draw [style={special_edge4}] (137.center) to (138.center);
		\draw [style={special_edge4}] (139.center) to (140.center);
		\draw [style={special_edge4}] (145.center) to (146.center);
		\draw [style={special_edge4}, bend left=90, looseness=0.75] (149.center) to (150.center);
		\draw [style={special_edge4}] (148.center) to (147.center);
		\draw [style={special_edge4}] (142.center) to (141.center);
		\draw [style={special_edge4}] (144.center) to (143.center);
	\end{pgfonlayer}
\end{tikzpicture}
\caption{Illustration of \cref{alg:canonical_network_recursive} when used to attach leaf 16 to $\Cr{\Y_1}$, where $\C$ is the canonical network from \cref{fig:central_edges} and $\Y_1 = \{1, \ldots, 15\}$ contains its first fifteen leaves. The algorithm recurses (from left to right, excluding the grey leaf 16) on the leaf sets $\Y_1$, $\Y_2 = \{1, 9, \ldots, 15\}$, $\Y_3 = \{1, 9, 10, 11, 14, 15\}$ and $\Y_4 = \{1, 9, 10, 14, 15\}$. In the recursive steps $i \in \{ 1, 2, 4 \}$, a central blob $\B_i$ (in pink) and an incident nice central edge $u_i v_i$ (the thick light/dark blue edge) are found. In recursive step $3$, the central blob $\B_3$ is undefined and the thick light blue edge $u_3 v_3$ instead originates from a nice spine edge (in dotted pink) of the previous central cycle~$\B_2$. In the recursive steps $i \in \{ 1, 2, 3 \}$, the light blue edge $u_i v_i$ is a pointing edge for leaf 16 (with the arrowhead showing the orientation) which is used to cut off part of the network. In the last recursive step, the dark blue edge $u_4 v_4$ is a strong stem edge used to attach leaf 16 (in grey). The algorithm subsequently moves back up to the top level of the recursion, each time gluing the relevant parts back to the canonical subnetworks of the network (from right to left, now including the grey leaf 16).}
\label{fig:alg_canonical}
\end{figure}

\begin{theorem}\label{thm:alg_canonical_network}
Given the quarnet-splits of a semi-directed level-1 network $\N$ on $\X = \{ x_1, \ldots, x_n \}$, there exists an algorithm that constructs the canonical form of $\N$ in $\bigO(n^2)$ time using $\bigO(n \log n)$ quarnet-splits of $\N$.
\end{theorem}
\begin{proof}
\emph{Correctness:} The algorithm starts with a single edge between two arbitrary leaves and then repeatedly attaches the remaining $\bigO (n)$ leaves in arbitrary order with \cref{alg:canonical_network_recursive}. Note that we do not provide the optional cut-edge in the top level recursive calls of \cref{alg:canonical_network_recursive}. To prove correctness it will be enough to prove the following claim by strong induction on $k$: given a set $\Y \subset \X$ of $k$ leaves, the canonical network $\Cr{\Y}$, some leaf $x \in \X \setminus \Y$, and optionally a non-trivial cut-edge $uv$, \cref{alg:canonical_network_recursive} constructs the canonical network $\Cr{\Y \cup \{x \}}$. In the base cases where $k \in \{2,3,4\}$, the claim follows directly from \cref{lem:unique_network_stem,cor:find_active_passive}. This is true since the canonical stem can never be a stem cycle in a canonical network with at most four leaves, so \cref{lem:unique_network_stem} covers all cases.

Now let $k\geq 5$ be arbitrary and assume that the claim holds for all values $k'$ with $2 \leq k' < k$. If no cut-edge $uv$ is provided, \cref{lem:nice_central_edge} will always find a central blob $\B$ and incident nice central edge $uv$. \cref{cor:find_active_passive} will correctly determine whether $uv$ is a strong stem edge, weak stem edge or pointing edge. If $uv$ is a strong or weak stem edge, \cref{lem:unique_network_stem,lem:find_passive_cutpath} prove correctness.

The if-statement on line~\ref{line:case3} is entered when $uv$ is a pointing edge with orientation $uv$. The glueing operation to obtain $\Cr{\Y \cup \{x\}}$ on line \ref{line:glue} is then correct by \cref{cor:glue_restrictions}. It remains to show that the creation of $\Cr{B \cup \{x, a\}}$ is correct. First, consider the case where $\B$ is undefined, $u$ is in $\B$ and/or $\B$ is a central vertex. If $\B$ is undefined, a non-trivial cut-edge was provided and we have $2 \leq |B \cup \{a\}| \leq k-1 < k$. Otherwise, if $u$ is in $\B$ and/or $\B$ is a central vertex, $uv$ is a nice central edge. Specifically, if $u$ is in $\B$ (so $uv$ points away from $\B$), we have by the definition of a central edge that $|B \cup \{a\}| \leq \frac12 k +1 < k$. On the other hand, if $v$ is in $\B$ and $\B$ is a central vertex, then $|B \cup \{a\}| \leq \frac34 k +1 < k$ by the definition of a nice central edge. In all cases we recurse on a canonical network with strictly fewer leaves (and more than two leaves), which shows correctness by the induction hypothesis.

For the other case, note that since $\B$ is not undefined, no cut-edge $uv$ could have been provided as an argument. Thus, because $\B$ is also not a central vertex, it is a central cycle and $uv$ is a nice central edge incident to it such that $v$ is the reticulation of $\B$. By \cref{lem:nice_central_edge}, we can indeed find a nice spine edge of the cycle $\B$, which in turn corresponds to a non-trivial cut-edge $st$ in $\Cr{B}$ (see e.g. \cref{fig:central_edges}). Thus, $2 \leq |B| < k$ and we recurse on a strictly smaller network (with at least two leaves), proving that $\Cr{B \cup \{x\}}$ is constructed correctly by the induction hypothesis. Since $\B$ is a central cycle and $v$ is the reticulation of $\B$, \cref{lem:close_cycle} then proves that $\Cr{B \cup \{x, a\}}$ is constructed correctly.

\emph{Complexity:} Recall that the main algorithm repeatedly invokes \cref{alg:canonical_network_recursive} (without providing a cut-edge at the top level of recursion) to attach all the $\bigO (n)$ leaves. Thus, for the theorem it will be enough to show that \cref{alg:canonical_network_recursive} takes $\bigO (k)$ time and uses $\bigO( \log k)$ quarnet-splits on a canonical network of $k$ leaves (when no cut-edge is provided at the top level of recursion) .

We will first consider the maximum recursion depth $D(k)$ of \cref{alg:canonical_network_recursive} applied to a canonical network of $k$ leaves. At the top level of recursion no cut-edge is provided and $\B$ is defined. Thus, when the algorithm recurses and $u$ is in $\B$ or $\B$ is a central vertex, it recurses on a canonical network with at most $\frac34 k + 1$ leaves (see the correctness proof). The other case when the algorithm recurses is if $v$ is in $\B$ and $\B$ is a central cycle. Then, again by the previous correctness proof, it recurses on at most $k-1$ leaves. However, a cut-edge derived from a nice spine edge with spine split $A'|B'$ is also provided for the next level of recursion. We then know that $2 \leq |A'| \leq \frac34 k$ and $2 \leq |B'| \leq \frac34 k$. But this means that the recursion on the next level is on a leaf set with at most $\frac34 k + 1$ leaves. Hence, we obtain the recurrence formula $D(k) \leq D \left(\frac{3}{4} k + 1\right) + 2$ for $k \geq 5$ with $D(5) \leq 1$. Whenever $k\geq 5$ we have $\frac34 k + 1 < \frac{19}{20} k$ and therefore the formula simplifies to $D(k) \leq D \left(\frac{19}{20} k \right) + 2$. It then easily follows that the recursion depth is $\bigO( \log k)$.

For the final complexity, note that the glueing operations, attaching leaves, finding splits and (implicitly) inducing canonical subnetworks all take no more than $\bigO(k)$ time per recursive level. This holds because the number of vertices and edges in a canonical network is a linear function of its leaves. All the other non-recursive operations also take $\bigO (k)$ time as outlined in \cref{cor:find_active_passive,lem:find_passive_cutpath,lem:nice_central_edge,lem:close_cycle}. Together with the recursion depth of $\bigO( \log k)$ this results in a total time complexity of $\bigO (k)$ for \cref{alg:canonical_network_recursive}. The operations from \cref{cor:find_active_passive,lem:nice_central_edge,lem:close_cycle} all use $\bigO (1)$ quarnet-splits, whereas \cref{lem:find_passive_cutpath} uses $\bigO (\log k)$ quarnet-splits. However, the method from \cref{lem:find_passive_cutpath} is only used when the algorithm does not recurse any further. Hence, it follows from the recursion depth that the total number of quarnet-splits used by \cref{alg:canonical_network_recursive} is $\bigO( \log k)$.
\end{proof}

\subsection{Reconstructing a level-1 network from its canonical form}\label{subsec:canonical_lev1}
The previous subsection was concluded with an algorithm that constructs the canonical form of a semi-directed level-1 network from its quarnet-splits. If more information is available, i.e. the full quarnets are known, we can create the complete network from its canonical form. This is described in the following theorem. Assuming that the network is triangle-free even allows us to construct the complete network while disregarding any triangles in the quarnets. As explained in the introductory section, this aligns nicely with recent work showing that the triangles in quarnets are hard to locate in practice \cite{holtgrefe2024squirrel,barton2022statistical,martin2023algebraic}. Furthermore, we also describe a variant of our algorithm that constructs most of the semi-directed level-1 network from its displayed quartets instead, aligning with the identifiability results from \cite{banos2019identifying}.

\begin{theorem}\label{thm:level1_construction}
Let $\N$ be a semi-directed level-1 network on $\X = \{ x_1, \ldots, x_n\}$ with $n\geq 4$. 
\begin{enumerate}[label={(\alph*)},noitemsep,topsep=0pt]
    \item Given the quarnets of $\N$, there exists an algorithm that reconstructs $\N$ in $\bigO(n^2)$ time using $\bigO(n \log n)$ quarnets of $\N$.
    \item Given the quarnet-splits and the four-cycle quarnets of $\N$, there exists an algorithm that reconstructs $\N$ up to placing the reticulations in its triangles and up to collapsing its triangles in $\bigO(n^2)$ time using $\bigO(n \log n)$ quarnet-splits and $\bigO(n)$ four-cycle quarnets of $\N$.
    \item Given the displayed quartets of $\N$, there exists an algorithm that reconstructs $\N$ up to placing the reticulations in its triangles and 4-cycles and up to collapsing its triangles in $\bigO(n^2)$ time using $\bigO(n \log n)$ displayed quartets of $\N$.
\end{enumerate}
\end{theorem}
\begin{proof}
We first prove part (a) of the theorem and treat parts (b) and (c) separately at the end of the proof. \cref{thm:alg_canonical_network} proves that we can construct the canonical network $\C$ in quadratic time from $\bigO(n \log n)$ quarnet-splits of $\N$, which can be formed from the quarnets of $\N$. It remains to show that we can transform every blob of $\C$ into the correct blob of $\N$, using the following three steps. The complexity bounds will then follow directly from the fact that $\N$ has $\bigO (n)$ blobs and that the described methods take no more than $\bigO (n)$ time (and $\bigO(1)$ quarnets) per blob. Given a cycle $\Cycle$, we will refer to the circular order of the subnetworks around $\Cycle$ and the placement of the reticulation as the \emph{orientation} of $\Cycle$. 

\emph{Step 1.} Every blob of $\C$ that consists of a single degree-4 vertex $v$ is a 4-cycle $\Cycle$ in $\N$. Suppose the partition of $\X$ induced by $v$ is $Y_1 | Y_2 | Y_3 | Y_4$ and choose one leaf $y_i$ from every $Y_i$. Then the quarnet $\N|_{\{y_1, y_2, y_3, y_4 \}}$ will be a four-cycle with one reticulation vertex and its orientation tells us exactly what the orientation of $\Cycle$ must be in~$\N$.

\emph{Step 2.} Similarly, every blob of $\C$ that consists of a single degree-3 vertex $v$ is either a single vertex $v'$ or a triangle $\Cycle$ in~$\N$. Suppose the partition of $\X$ induced by $v$ is $Y_1 | Y_2 | Y_3 $ such that $|Y_3|\geq 2$ (this exists since $n\geq 4$). Choose one leaf $y_i$ from every $Y_i$ and choose $y_4$ as a different leaf from $Y_3$. The quarnet $\N|_{\{y_1, y_2, y_3, y_4 \}}$ will then be a quartet tree, a single triangle, or a double triangle (see e.g. \cref{fig:quarnets}). If we take the subnetwork of this quarnet induced by $\{y_1, y_2, y_3\}$, we either get a triangle or a 3-star. If this \emph{trinet} is a triangle, we know the exact orientation of the triangle $\Cycle$ in $\N$. If the trinet is a 3-star, we know that $v$ is a single vertex $v'$ in $\N$.

\emph{Step 3.} Lastly, in every large cycle $\Cycle$ of $\N$ the two cycle contraction edges are contracted when obtaining $\C$. A similar strategy as above allows us to determine how to undo this operation. Specifically, suppose that $Y_1 | Y_2 | \ldots |Y_{k-1} | Y_k| Z$ is the partition of $\X$ induced by $\Cycle$ with $Z$ below the reticulation, and $Y_1, Y_2$ and $Y_{k-1}, Y_k$ the leaf sets corresponding to the two cycle contraction edges. If we let $y_i \in Y_i$ and $z \in Z$ be arbitrary, then the two four-cycle quarnets $\N|_{\{z, y_1, y_2, y_{k-1}\}}$ and $\N|_{\{z, y_{k-1}, y_k, y_1 \}}$ directly show what the orientation of $\Cycle$ must be in $\N$.

For part (b) of the theorem it is enough to note that Step 2 is redundant if we want to reconstruct $\N$ up to collapsing its triangles and up to placing reticulations in them. For part (c), we use that from the displayed quartets of a semi-directed network we can obtain the quarnet-splits and the circular ordering of the four-cycle quarnets without the placement of the reticulation (see \cite[Lem.\,5.1]{rhodes2024identifying}). To see the result, note that we can again skip Step~2, whereas we can perform most of Step~1 except for placing the reticulation in the 4-cycles of the network. Finally, note that the number of displayed quartets used is at most three times the corresponding number of quarnets used.
\end{proof}

The proof of the previous result also implicitly shows why the canonical network is the most refined network one can unambiguously construct from quarnet-splits. In particular, it is impossible to locate triangles using quarnet-splits, to infer the circular order of subnetworks and the reticulation vertex of 4-cycles, and to determine how the two cut-edge pairs next to a reticulation vertex of a large cycle are ordered. Therefore, quarnet-splits do not uniquely determine a complete semi-directed level-1 network.

Recall that \cref{prop:lower_bound} showed that any algorithm using only quarnet-splits to reconstruct the tree-of-blobs of a semi-directed level-$\ell$ network needs to use $\Omega (n \log n + \ell \cdot n)$ quarnet-splits. Since one can easily create the tree-of-blobs of a level-1 network from its canonical form, \cref{thm:alg_canonical_network} implies that one can reconstruct the tree-of-blobs of a level-$1$ network from $\bigO (n \log n)$ quarnet-splits. Thus, the lower bound from \cref{prop:lower_bound} is tight for level-1 networks. Furthermore, a similar information-theoretic argument as in \cref{prop:lower_bound} shows that the algorithm from \cref{thm:level1_construction} is optimal in terms of the number of quarnets it uses. That is, no algorithm exists that uses asymptotically less than $\bigO( n \log n)$ quarnets to reconstruct a semi-directed level-1 network. We present the formal statement in the following proposition. The argument can trivially be adapted to show that the number of displayed quartets used in part (c) of the previous theorem is also asymptotically optimal.

\begin{proposition}\label{prop:lower_bound2}
Given the quarnets of a semi-directed level-$1$ network $\N$ on $\X = \{x_1, \ldots, x_n\}$, any algorithm using only quarnets to reconstruct $\N$ needs to use $\Omega (n \log n)$ quarnets.
\end{proposition}
\begin{proof}
In \cite{kannan1996determining} it is shown that there are $2^{\Omega (n \log n)}$ possible binary phylogenetic trees on $n$ leaves, which is thus surely a lower bound on the number of $n$-leaf semi-directed level-1 networks. Up to labeling the leaves, there are only six possible level-1 quarnets (see \cref{fig:quarnets}). Counting the number of different labelings then reveals that on any given set of four leaves there are 3 possible quartet trees, 18 single triangles, 18 double triangles, and 12 four-cycles. Thus, by an information-theoretic argument, any algorithm using only quarnets to reconstruct $\N$ needs to use $\log_{3+18+18+12} (2^{\Omega(n \log n)}) = \Omega (n \log n)$ quarnets in the worst case.
\end{proof}

\section{Discussion}\label{sec:discussion}
The main contributions of this paper are two-fold. First, we presented an $\bigO(n^3)$ time algorithm that reconstructs the tree-of-blobs of any binary $n$-leaf semi-directed network with unbounded level, using $\bigO(n^3)$ splits of its quarnets (assuming direct access to the splits of all quarnets). We have not shown that this is optimal but have shown that any such algorithm needs $\Omega (n^2)$ quarnet-splits in the worst case. Secondly, we created an algorithm that reconstructs binary $n$-leaf semi-directed level-1 networks in $\bigO( n^2)$ time, using an optimal number of $\bigO(n \log n)$ quarnets (assuming direct access to all quarnets). A variant of this algorithm can also reconstruct most of the semi-directed level-1 network using its displayed quartets instead.

Two obvious open questions that remain are whether the $\Theta (n)$ gap for the tree-of-blobs reconstruction algorithm can be bridged and whether the time complexity of the level-1 network reconstruction algorithm can be reduced from $\bigO (n^2)$ to $\bigO( n \log n)$. For the latter improvement, one might draw some inspiration from \cite{brodal2001complexity}, who present an optimal algorithm that reconstructs undirected phylogenetic trees of degree $d$ from $\Theta (d\cdot n \cdot \log_d n)$ quartets and running in similar time.\footnote{This result follows from the equivalence between undirected tree reconstruction from quartets and directed tree reconstruction from triplets~\cite{lingas1999efficient}.} The optimality of that algorithm also signifies an inherent difference between the seemingly similar tree-of-blobs reconstruction from quarnet-splits and the phylogenetic tree reconstruction from quartets. Whereas the number of quartets used in the algorithm by \cite{brodal2001complexity} reduces to $\Theta (n \log n)$ for binary trees, we showed that reconstructing the tree-of-blobs of any binary semi-directed network requires $\Omega (n^2)$ quarnet-splits. One might suspect this difference follows from the fact that the trees-of-blobs themselves could be of high degree. This is not true however, since we showed that trees-of-blobs of binary semi-directed level-1 networks can be reconstructed with $\bigO (n \log n)$ quarnet-splits, even if they have high-degree blobs.

A further research direction related to this paper is the step from level-1 to level-2. In \cite{iersel2022algorithm} this was achieved for the construction of rooted/directed level-2 phylogenetic networks from trinets (3-leaf subnetworks). Recent findings by Huber et al. \cite{huber2024splits} show that semi-directed level-2 networks can be distinguished by their quarnets. This suggests the jump from level-1 to level-2 is possible for the semi-directed case. We already know how to reconstruct the tree-of-blobs of a level-2 network (although this could perhaps be sped up) which only leaves the separate blobs to be reconstructed. To this end, it might be worthwhile to generalize our canonical form to higher-level networks, thus formalizing the most refined level-2 network that can be reconstructed from quarnet-splits. Note that \cite{ardiyansyah2021distinguishing} has already made some progress on identifiability results for the level-2 case under group-based models of evolution.

Lastly, from a practical point of view, a more robust algorithm that also works for imperfect data would be a logical next step. One objective for such a method might be to find a semi-directed level-1 network that induces as many given quarnets as possible: the \emph{maximum quarnet compatibility} problem. However, this is NP-hard, which can easily be shown using the fact that it already is for trees and quartets \cite{bryant2001constructing}. Thus, there is a need for heuristics, ideally ones that are \emph{consistent}: returning the true network whenever the data is perfect. The software tool \textsc{Squirrel} has taken a first step in this direction by allowing a \emph{dense} set of quarnets (i.e. exactly one quarnet for each set of four leaves) that do not necessarily come from a single network \cite{holtgrefe2024squirrel}. In light of our work, a next step could be to allow a \emph{non-dense} set of quarnets (i.e. at most one quarnet for each set of four leaves) as input, where the quarnets are not necessarily all induced by the same network. Allowing fewer quarnets in the input could lead to a significant speed-up. A sensible starting point is to investigate methods that construct rooted/directed level-1 phylogenetic networks from 
a non-dense set of subtrees on few leaves (e.g. \textsc{Lev1athan} \cite{huber2010practical}).

\section*{Acknowledgements}
We thank the reviewers for their helpful comments and suggestions to improve the paper.

\bibliographystyle{elsarticle-num} 
\bibliography{references}

\appendix 

\section{Proof of Lemma 1}
\label{sec:appendix}

\splitlemma*
\begin{proof}
As discussed in \cref{sec:splits}, the proof of this lemma is along the lines of the proof of a slightly weaker result by Huber~et~al.~\cite{huber2024splits}. In particular, they show that $A|B$ is a split in $\N$ if and only if $a_1 a_2 | b_1 b_2$ is a quarnet-split of $\N$ for all $a_1, a_2 \in A$ and $b_1, b_2 \in B$. The first direction of our lemma follows directly from that result. As in \cite{huber2024splits}, we prove the reverse direction by induction on the number of non-trivial splits in $\N$. 

For the base case, assume that $\N$ is simple and let $A|B$ be any non-trivial partition of $\X$. The proof by Huber~et~al.~\cite{huber2024splits} shows that in that case there exist $\abar_1 \in A$ and $\bbar_1 \in B$ such that $\abar_1 \abar_2 | \bbar_1 \bbar_2$ is not a quarnet-split for any $\abar_2 \in A \setminus \{\abar_1 \}$ and $\bbar_2 \in B \setminus \{ \bbar_1\}$. Our base case follows directly from this statement.

It remains to do the induction step. First, let $A|B$ be any non-trivial partition of $\X$, fix some $a_1 \in A, b_1 \in B$ and assume that $a_1 a_2 | b_1 b_2$ is a quarnet-split for all $a_2 \in A \setminus \{a_1\}$ and $b_2 \in B \setminus \{b_2\}$ [Assumption 1]. Furthermore, suppose that $\N$ is not simple and that there exists a non-trivial split $C|D$. Towards a contradiction, assume that $A \not\subseteq C$, $A \not\subseteq D$, $B \not\subseteq C$, and $B \not\subseteq D$. Without loss of generality, suppose that $a_1 \in C$ and $b_1 \in D$. (The other cases are analogous.) Since $A \not\subseteq C$, there exists an $\abar\notin C$, so $\abar \in D$. Similarly, since $B \not\subseteq D$, there exists a $\bbar \notin D$ such that $\bbar \in C$. Note that we must have that $\abar \neq a_1$ and $\bbar \neq b_1$. Thus, we have that $a_1, \bbar \in C$ and $b_1, \abar \in D$. Since $C|D$ is a split in $\N$, the forward direction of the lemma says that then $a_1 \bbar | \abar b_1$ is a quarnet-split of $\N$. This means that $a_1 \abar | b_1 \bbar$ is not a quarnet-split of $\N$: a contradiction with Assumption~1. All in all, we must have that $A \subseteq C$, $A \subseteq D$, $B \subseteq C$, or $B \subseteq D$.

Without loss of generality, assume that $B\subseteq D$, and consequently $C \subseteq A$. Let $uv$ be the cut-edge in $\N$ inducing the split $C|D$. Then, its removal from $\N$ creates two connected components: $\N_C$ containing $u$ and all leaves from~$C$, and $\N_D$ containing $v$ and all leaves from $D$. We now construct a new network $\N'$ from $\N$ by identifying all vertices of $\N$ that are in $\N_C$ as a single leaf $c^*$. Then, $\N'$ has one non-trivial split less than $\N$. Furthermore, if we let $A' = (A \setminus C) \cup \{ c^* \}$ and $B' = B$, then $A'|B'$ is a partition of the leaves of $\N'$. Distinguishing between two cases, we now show that $A'|B'$ is a split in $\N'$. By the construction of $\N'$, this means that $A|B$ is a split in $\N$, proving the reverse direction.

\emph{Case 1:} $a_1 \in C$. Let $\abar_2 \in A' \setminus \{c^*\}$ and $\bbar_2 \in B' \setminus \{b_1 \}$ be arbitrary. Then, $\abar_2 \in (A \setminus C) \setminus \{a_1\}$ and $\bbar_2 \in B \setminus \{b_1\}$. Furthermore, $\abar_2 \neq a_1$ by definition of $A'$. Assumption 1 then ensures that $a_1 \abar_2 | b_1 \bbar_2$ is a quarnet-split of~$\N$. Since $a_1 \in C$, we then have $c^* \abar_2 | b_1 \bbar_2$ as a quarnet-split of $\N'$. Because $\abar_2 \in A' \setminus \{c^*\}$ and $\bbar_2 \in B' \setminus \{b_1 \}$ were arbitrary, the induction hypothesis (and the fact that $\N'$ has a non-trivial split less than $\N$) shows that $A'|B'$ is a split in $\N'$.

\emph{Case 2:} $a_1 \notin C$. Let $\abar_2 \in A' \setminus \{a_1\}$ and $\bbar_2 \in B' \setminus \{b_1 \}$ be arbitrary. Clearly, $\bbar_2 \in B \setminus \{ b_1\}$. By Assumption~1, $a_1 c | b_1 \bbar_2$ is a quarnet-split of $\N$ for all $c \in C \subseteq A \setminus \{a_1\}$. Consequently, if $\abar_2 = c^*$, then $a_1 \abar_2 | b_1 \bbar_2$ is a quarnet-split of $\N'$. Otherwise, if $\abar_2 \in (A \setminus C ) \setminus \{a_1\}$, Assumption 1 implies that $a_1 \abar_2 | b_1 \bbar_2$ is a quarnet-split of~$\N$ and thus of $\N'$. Because $\abar_2 \in A' \setminus \{a_1\}$ and $\bbar_2 \in B' \setminus \{b_1 \}$ were arbitrary, the induction hypothesis (and the fact that $\N'$ has a non-trivial split less than $\N$) shows that $A'|B'$ is a split in $\N'$.
\end{proof}

\end{document}